\documentclass[conference]{IEEEtran}
\IEEEoverridecommandlockouts  

\usepackage[british]{babel}
\usepackage[T1]{fontenc}
\usepackage[utf8x]{inputenc}

\DeclareFontFamily{U}{mathx}{\hyphenchar\font45}
\DeclareFontShape{U}{mathx}{m}{n}{
      <5> <6> <7> <8> <9> <10>
      <10.95> <12> <14.4> <17.28> <20.74> <24.88>
      mathx10
      }{}
\DeclareSymbolFont{mathx}{U}{mathx}{m}{n}
\DeclareFontSubstitution{U}{mathx}{m}{n}
\DeclareMathAccent{\widecheck}{0}{mathx}{"71}

\usepackage{amsmath,amsthm,amssymb,amsfonts}
\usepackage{amsbsy}    
\usepackage{stmaryrd}  
\usepackage{calrsfs}                             
\DeclareMathAlphabet{\pazocal}{OMS}{zplm}{m}{n}  
\usepackage[mathscr]{eucal}                      
\usepackage{subcaption,graphicx,wrapfig}
\usepackage{xfrac,xspace,xifthen}
\usepackage[dvipsnames,table]{xcolor}
\usepackage[inline]{enumitem}
\usepackage{booktabs,multirow,array}
\usepackage{rotating}  
\usepackage{pdflscape,afterpage}
\usepackage[normalem]{ulem}
\usepackage{contour}
\usepackage[vlined]{algorithm2e} 
\usepackage{tikz,tikzsymbols,pgfplots}
\usepackage{tokcycle,marginnote}
\usepackage{twoopt}  
\usepackage{relsize}
\usepackage{breakcites}
\usepackage{lipsum}
\usepackage[%
	breaklinks=true,%
	colorlinks=true,%
	allcolors=Blue,%
	urlcolor=black!60!blue!95]{hyperref}
\usepackage[nameinlink]{cleveref}

\newboolean{tosubmit}
\setboolean{tosubmit}{true}

\newboolean{showhelpers}
\setboolean{showhelpers}{false}

\newboolean{forCSF}
\setboolean{forCSF}{false}

\makeatletter
\newtheorem*{rep@theorem}{\rep@title}
\newcommand{\newreptheorem}[2]{%
\newenvironment{rep#1}[1]{%
 \def\rep@title{#2 \ref{##1}}%
 \begin{rep@theorem}}%
 {\end{rep@theorem}}}
\makeatother

\theoremstyle{plain}
\newtheorem{theorem}{Theorem}
\newreptheorem{theorem}{Theorem}  
\newtheorem{lemma}{Lemma}
\newreptheorem{lemma}{Lemma}  
\newtheorem{corollary}{Corollary}
\newtheorem{proposition}{Proposition}
\theoremstyle{definition}
\newtheorem{definition}{Definition}
\newtheorem{example}{Example}

\graphicspath{{imgs/},}  

\makeatletter

\def\orcidID#1{\href{http://orcid.org/#1}{\protect\raisebox{-1.25pt}{\protect\includegraphics[height=8pt]{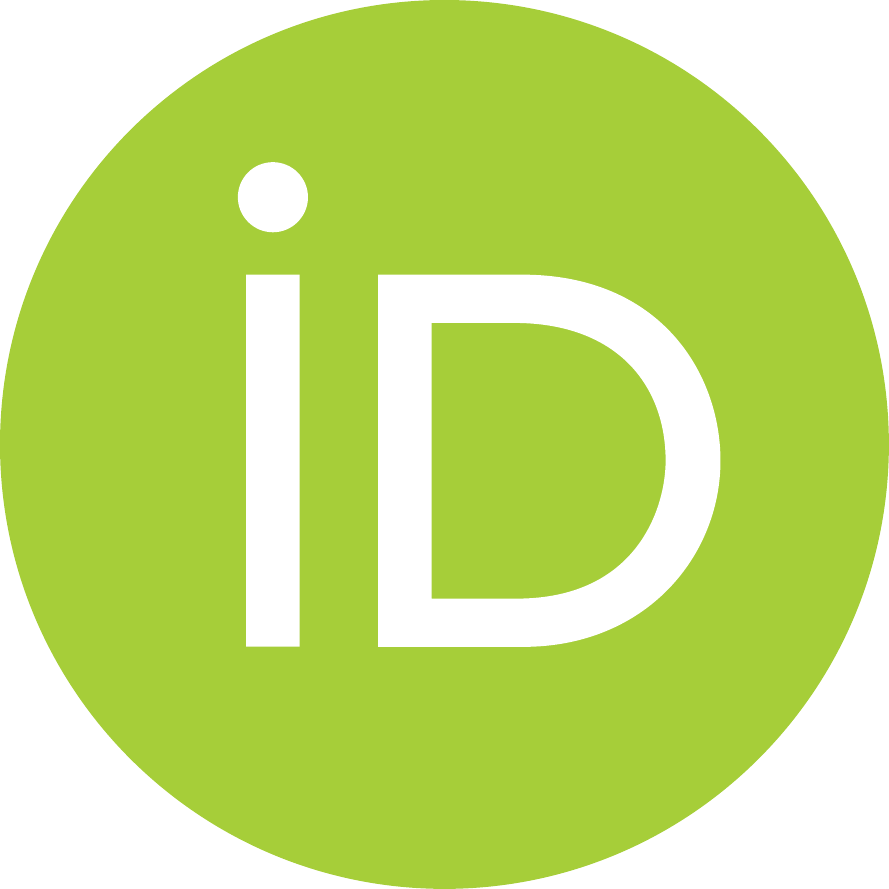}}}}
\makeatother

\usetikzlibrary{arrows,arrows.meta,calc,shapes,positioning,automata}
\tikzstyle{every picture}+=[remember picture, overlay, anchor=center]
\tikzstyle{every node}=[initial text=]
\tikzstyle{every picture}=[>=stealth]
\tikzstyle{BE}=[ellipse, draw, fill=normalfill, minimum width=7mm, inner sep=1]

\Crefname{figure}{Fig.}{Figs.}
\crefname{figure}{figure}{figures}
\Crefname{tabular}{Table}{Tables}
\Crefname{algorithm}{Algo.}{Algos.}
\crefname{algorithm}{algorithm}{algorithms}
\Crefname{definition}{Def.}{Defs.}
\crefname{definition}{definition}{definitions}
\Crefname{theorem}{Theo.}{Theorems}
\crefname{theorem}{theorem}{theorems}
\Crefname{lemma}{Lemma}{Lemmas}
\Crefname{corollary}{Cor.}{Corollaries}
\crefname{corollary}{corollary}{corollaries}
\Crefname{proposition}{Prop.}{Propositions}
\crefname{proposition}{proposition}{propositions}
\Crefname{example}{Example}{Examples}
\Crefname{section}{Sec.}{Secs.}
\crefname{section}{section}{sections}
\Crefname{equation}{eq.}{eqs.}
\crefname{equation}{equation}{equations}

\SetAlCapFnt{\mdseries}  
\SetAlCapSkip{\medskipamount}
\SetStartEndCondition{\ \ }{\ \ }{}

\SetCommentSty{algocomment}

\contourlength{.9pt}

\makeatletter
\renewcommand{\paragraph}{\@startsection{paragraph}{5}{0em}%
  {.7ex plus .2ex minus .1ex}%
  {-.5em}%
  {\bfseries}}
\makeatother

\newcommand{\hlbox}[1]{%
  \smallskip\begin{center}
  \fboxrule1pt\fboxsep3pt\fcolorbox{black!45}{black!8}{%
  \begin{minipage}{.96\linewidth}#1\end{minipage}}
  \end{center}\smallskip}

\ifthenelse{\boolean{showhelpers}}{%
  \usepackage[pagewise,switch,modulo]{lineno}  
  \linenumbers
  
}{}

\ifthenelse{\boolean{tosubmit}}{}{ 
  \usepackage[firstpage]{draftwatermark}
}

\makeatletter
\def\THICKhrulefill{\leavevmode \leaders \hrule height 5pt\hfill \kern \z@}
\def\getfirst#1#2\relax{\tctestifnum{\count@stringtoks{#1}>1}{ERROR}{#1}}
\makeatother
\newcommand{\colorpar}[3]{\colorbox{#1}{\parbox{#2}{#3}}}
\newcommand{\marginremark}[3]{%
  \ifthenelse{\boolean{tosubmit}}{}{
	\marginnote{\colorpar{#2}{1.2\linewidth}{\raggedright\color{#1}#3}}
}}
\newcommand{\highlightedremark}[4]{%
  \ifthenelse{\boolean{tosubmit}}{}{
	\begin{center}\fcolorbox{#1}{#2}{%
	\begin{minipage}{.98\linewidth}\color{#1}%
	\textbf{\THICKhrulefill[ #3 ]\THICKhrulefill}%
	\par\noindent#4\end{minipage}}\end{center}%
}}
\newcommand{\hey}[4]{%
  \ifthenelse{\boolean{tosubmit}}{}{
  \leavevmode\marginnote{\sffamily\Large\color{#1}@\getfirst#3\relax\relax}
  \colorbox{#2}{\sffamily\bfseries{@#3:}}~{\sffamily\color{#1}#4}}}
\newcommand{\todo}[1]{%
  \ifthenelse{\boolean{tosubmit}}{}{
  \noindent\textsf{\color{Red}\textbf{TODO:} #1}%
  \begingroup\renewcommand*{\marginnotevadjust}{-3ex}%
  \marginnote{\textsf{\color{red}\bfseries TODO}}\endgroup}}
\newcommand{\tocite}[1][??]{%
  \ifthenelse{\boolean{tosubmit}}{}{
  \noindent\textbf{\sffamily\textcolor{blue!85}{[#1]}}%
  \begingroup\renewcommand*{\marginnotevadjust}{-3ex}%
  \marginnote{\textsf{\color{blue}\bfseries CITE!}}\endgroup}}
\colorlet{MS-fg}{WildStrawberry}
\colorlet{MS-bg}{Plum!6}
\colorlet{CEB-fg}{TealBlue!75!green!75!black}
\colorlet{CEB-bg}{Aquamarine!8}

\newcommand{\rmkCEB}[1]{\marginremark{CEB-fg}{CEB-bg}{\tiny\sffamily[CEB]\ #1}}

\newcommand{\hrmkCEB}[1]{\highlightedremark{CEB-fg}{CEB-bg}{CEB}{#1}}

\let\oldtop\top
\let\oldbot\bot
\renewcommand{\top}{\mathtt{{1}}}
\renewcommand{\bot}{\mathtt{{0}}}
\renewcommand{\emptyset}{\varnothing}
\renewcommand{\vec}{\mathaccent "017E\relax}  
\renewcommand{\restriction}{\mathord{\upharpoonright}} 
\newcommand\restr[2]{{
  \left.\kern-\nulldelimiterspace #1 \vphantom{\big|} \right|_{#2}}}
\newcommand{\BB}{\ensuremath{\mathbb{B}}\xspace}  
\newcommand{\NN}{\ensuremath{\mathbb{N}}\xspace}  
\newcommand{\QQ}{\ensuremath{\mathbb{Q}}\xspace}  
\newcommand{\RR}{\ensuremath{\mathbb{R}}\xspace}  
\newcommand{\from}{\colon}
\newcommand{\minb}{\mathbin{\mathrm{min}}}
\newcommand{\maxb}{\mathbin{\mathrm{max}}}
\newcommand{\bfx}{\ensuremath{\boldsymbol{x}}\xspace}
\newcommand{\card}[1]{\ensuremath{\vert{#1}\vert}\xspace}
\newcommand{\bigO}[1]{\ensuremath{O{\left(#1\right)}}\xspace}
\newcommand{\overbar}[1]{\mkern 1.8mu\overline{\mkern-1.8mu#1\mkern-2.4mu}\mkern 2.4mu}
\newcommand{\acronym}[1]{\ensuremath{\text{\larger\scshape#1}}\xspace}
\newcommand{\mathobject}[1]{\ensuremath{\text{\scalebox{.92}{$#1$}}}}
\newcommand{\bfred}[1]{\textcolor{Red}{\bfseries{#1}}\xspace}
\newcommand{\Vars}{\ensuremath{\mathit{Vars}}\xspace}
\DeclareMathOperator{\bddOp}{\mathobject{B}}
\DeclareMathOperator{\BDDlab}{\mathit{Lab}}
\DeclareMathOperator{\low}{\mathit{Low}}
\DeclareMathOperator{\high}{\mathit{High}}
\newcommand{\Root}{\mathobject{R}}
\newcommand{\BDDnodes}{\mathobject{W}\xspace}
\newcommand{\BDDnodesN}{\ensuremath{\BDDnodes_{\mkern-5mu n \mkern 1mu}}\xspace}
\newcommand{\BDDnodesT}{\ensuremath{\BDDnodes_{\mkern-4mu t \mkern 1mu}}\xspace}
\newcommand{\BDDroot}[1][\bddOp]{\ensuremath{\Root_{#1}}\xspace}

\makeatletter  
\newcommand{\precneq}{\mathrel{\text{\prec@eq}}}
\newcommand{\prec@eq}{%
  \oalign{%
    \hidewidth$\m@th\prec$\hidewidth\cr
    \noalign{\nointerlineskip\kern1ex}%
    $\m@th\smash{\raisebox{0.65ex}{\rotatebox{90}{%
      \scalebox{1.1}[-1.1]{$\nshortmid$}}}}$\cr
    \noalign{\nointerlineskip\kern-.5ex}%
}}
\makeatother
\newcommand{\LINTIME}{\acronym{lintime}}     
\newcommand{\EXPTIME}{\acronym{exptime}}     

\newcommand{\PSPACE}{\acronym{pspace}}
\newcommand{\AT}{\acronym{at}}               
\newcommand{\ATs}{\acronym{at}{s}\xspace}
\newcommand{\SAT}{\acronym{sat}}             
\newcommand{\SATs}{\acronym{sat}{s}\xspace}
\newcommand{\DAT}{\acronym{dat}}             
\newcommand{\DATs}{\acronym{dat}{s}\xspace}

\newcommand{\DAG}{\acronym{dag}}             
\newcommand{\DAGs}{\acronym{dag}{s}\xspace}
\newcommand{\BDD}{\acronym{bdd}}             
\newcommand{\BDDs}{\acronym{bdd}{s}\xspace}
\DeclareMathOperator{\pos}{\mathit{pos}}
\newcommand{\BU}[1][]{\ensuremath{\mathtt{BU%
  \ifthenelse{\isempty{#1}}{}{_{\mkern1mu#1}}}}\xspace}
\newcommand{\BUSAT}{\BU[SAT]}
\newcommand{\BUDAT}{\BU[DAT]}
\newcommand{\BUBDD}{\ensuremath{\mathtt{BDD_{DAG}}}\xspace} 
\DeclareMathOperator{\kshortest}{\mathtt{shortest\_paths}}
\DeclareMathOperator{\iswellformed}{\mathtt{is\_well\_formed}}
\newcommand{\nodeType}[1]{\ensuremath{\mathtt{#1}}\xspace}
\newcommand{\tBAS}{\nodeType{BAS}}
\newcommand{\tOR}{\nodeType{OR}}
\newcommand{\tAND}{\nodeType{AND}}
\newcommand{\tSAND}{\nodeType{SAND}}
\newcommand{\allATs}{\ensuremath{\mathcal{T}}\xspace}
\DeclareMathOperator{\typOp}{\mathit{t}}       
\DeclareMathOperator{\chOp}{\mathit{ch}}       
\newcommand{\type}[1]{\ensuremath{\typOp({#1})}\xspace}
\newcommand{\child}[1]{\ensuremath{\chOp({#1})}\xspace}

\DeclareMathOperator{\desc}{\BAS}
\DeclareMathOperator{\sfun}{\mathit{f}}        
\newcommand{\sfunT}[1][\T]{\ensuremath{\sfun_{\!#1}}\xspace}
\newcommand{\sem}[1]{\ensuremath{\llbracket{#1}\rrbracket}}  
\newcommand{\ssem}[1]{\sem{#1}\xspace}                       
\newcommand{\dsem}[1]{\ssem{#1}}                             
\DeclareMathOperator{\ogOp}{\mathobject{G}}                  
\newcommand{\ograph}[1]{\ensuremath{\ogOp_{#1}}\xspace}
\newcommand{\before}[1][]{\mathbin{%
  \ifthenelse{\isempty{#1}}{%
    \tikz[baseline=-.6ex]{\draw[->,thin,x=1ex,y=1ex](0,0)--(1.7,0);}%
  }{%
    \tikz[baseline=-.48ex]{\draw[->,thin,x=1ex,y=1ex](0,0)--(1.7,0);%
    \node[x=1ex,y=1ex,inner sep=0pt]at(.75,.75){$\scriptscriptstyle{#1}$};}}}
}
\newcommand{\attack}[1][A]{\mathobject{#1}\xspace}
\newcommand{\suite}[1][S]{\ensuremath{\text{\relsize{-.5}$\pazocal{#1}$}}\xspace}
\newcommand{\allAttacks}{\ensuremath{\mathcal{A}}\xspace}
\newcommand{\allSuites}{\ensuremath{\rotatebox[origin=c]{-15}{$\mathscr{S}$\!}}\xspace}
\newcommand{\amin}[1][\T]{\sem{#1}\xspace}
\newcommandtwoopt{\poset}[2][\attack][\prec]{%
  \ensuremath{\langle{#1},{#2}\rangle}\xspace}
\newcommandtwoopt{\Hasse}[2][\attack][\prec]{%
  \ensuremath{\mathobject{H}_{\mkern-2mu#1}^{#2}}\xspace}
\newcommand{\ccomp}{\mathobject{C}\xspace}

\DeclareMathOperator{\attrOp}{\alpha}                     
\DeclareMathOperator{\metrSOp}{
  \mkern-1.2mu\vec{\mkern1.2mu\attrOp\mkern-1.2mu}\mkern1.2mu}
\DeclareMathOperator{\metrAOp}{\widehat{\attrOp}}         
\DeclareMathOperator{\metrOp}{\widecheck{\attrOp}}        
\newcommand{\attr}[1]{\ensuremath{\attrOp(#1)}\xspace}    
\newcommand{\metrA}[1]{\ensuremath{\metrAOp(#1)}\xspace}  
\newcommand{\metr}[1]{\ensuremath{\metrOp(#1)}\xspace}    
\newcommand{\Vdom}{\mathobject{V}\xspace}
\newcommand{\domain}{\mathobject{D}\xspace}
\newcommand{\ndomain}{\ensuremath{\domain_{\mkern-1mu\star}}\xspace}
\newcommand{\operOR}{\mathbin{\triangledown}}
\newcommand{\operAND}{\mathbin{\vartriangle}}
\newcommand{\operSAND}{\mathbin{\vartriangleright}}
\DeclareMathOperator*{\bigoperOR}{\bigtriangledown}
\DeclareMathOperator*{\bigoperAND}{\bigtriangleup}
\DeclareMathOperator*{\bigoperSAND}{\mathbin{\text{\raisebox{-.7ex}{\rotatebox{90}{$\bigtriangledown$}}}}\vphantom{\bigtriangledown}}
\newcommand{\neutral}[1]{\ensuremath{1_{\!#1}}\xspace}

\newcommand{\ntOR}{\neutral{\operOR}}
\newcommand{\ntAND}{\neutral{\operAND}}
\newcommand{\ATnodes}{\mathobject{N}\xspace}
\newcommand{\ATroot}[1][\T]{\ensuremath{\Root_{#1}}\xspace}
\newcommand{\TLA}{\acronym{tla}}             
\newcommand{\OR}{\acronym{or}}               
\newcommand{\AND}{\acronym{and}}             
\newcommand{\SAND}{\acronym{sand}}           
\newcommand{\BAS}{\acronym{bas}}             
\newcommand{\BASs}{\acronym{bas}{s}\xspace}
\newcommand{\T}[1][]{\ensuremath{\mathobject{T}\ifthenelse{\isempty{#1}}{}{_{\hspace{-2pt}#1}}}\xspace}
\newcommand{\sampleTs}{\ensuremath{\T[\text{\larger[.5]{$s$}}]}\xspace}
\newcommand{\sampleTd}{\ensuremath{\T[\text{\larger[.5]{$d$}}]}\xspace}
\newcommand{\ff}{\ensuremath{\mathit{ff}}\xspace}
\newcommand{\ww}{\ensuremath{\mathit{w}}\xspace}
\newcommand{\cc}{\ensuremath{\mathit{cc}}\xspace}
\DeclareMathOperator{\logicformula}{\mathobject{L}}
\newcommand{\LT}[1][\T]{\ensuremath{\logicformula_{#1}}\xspace}
\newcommand{\bddT}[1][\T]{\ensuremath{\bddOp_{#1}}\xspace}

\newcommand{\bddf}{\bddT[\!f]}

\def\TITLE{Efficient Algorithms for\\ Quantitative Attack Tree Analysis}

\hypersetup{pdftitle={\TITLE}}

\begin{document}
\title{%
	\TITLE
	\thanks{%
		We thank Sebastiaan Joosten for his help with the proof of
		\Cref{theo:NP_hard}; also Lars Kuijpers and Jarik Karsten for
		collaborations that led to our definition of ordering graphs and
		\Cref{alg:shortest_path_BDD} resp.
	This work was partially funded by NWO project~15474 (\emph{SEQUOIA}),
	and ERC Consolidator Grant 864075 (\emph{CAESAR}).%
	}
}
\author{
	\IEEEauthorblockN{%
		Carlos E.\ Budde\IEEEauthorrefmark{1}\orcidID{0000-0001-8807-1548}
		\qquad
		Mari\"elle Stoelinga\IEEEauthorrefmark{1}\IEEEauthorrefmark{2}\orcidID{0000-0001-6793-8165}
	}
	\IEEEauthorblockA{%
		\IEEEauthorrefmark{1}{University of Twente,
			Formal Methods and Tools, Enschede, the Netherlands.}\\
		\IEEEauthorrefmark{2}{Radboud University,
			Department of Software Science, Nijmegen, the Netherlands.}\\
		{\texttt{\small\{c.e.budde,m.i.a.stoelinga\}@utwente.nl}}
	}
}
\maketitle
\thispagestyle{plain}  
\pagestyle{plain}
\begin{abstract}
Numerous analysis methods for quantitative attack tree analysis have been proposed.
These algorithms compute relevant security metrics, i.e.\ performance indicators that quantify how good the security of a system is, such as the most likely attack, the cheapest, or the most damaging one.
This paper classifies attack trees in two dimensions: proper trees vs.\ directed acyclic graphs (i.e.\ with shared subtrees); and static vs.\ dynamic gates.
For each class, we propose novel algorithms that work over a generic attribute domain, encompassing a large number of concrete security metrics defined on the attack tree semantics.
We also analyse the computational complexity of our methods.
\end{abstract}


\section{Introduction}
\label{sec:intro}


Attack trees are a popular method in decision making for security, supporting the identification, documentation and analysis of cyberattacks.
They are part of many system engineering frameworks, e.g.\ \emph{UMLsec} \cite{RA15} and \emph{SysMLsec} \cite{AR13}, and are supported by industrial tools such as Isograph's \emph{AttackTree} \cite{IsographAT}. 

An attack tree (\AT) is a hierarchical diagram to systematically map potential  attack scenarios of a system, see \Cref{fig:AT_nodes,fig:AT}.
The root at the top of the diagram models the attacker's goal, which is further refined into subgoals by means of gates:
an \AND gate indicates that an attack is successful iff all children attacks succeed;
an \OR gate indicates that any single child suffices.
The leaves of the tree are basic attack steps (\BAS), which model indivisible actions such as cutting a wire.  

\paragraph{Static vs.\ dynamic attack trees}
Extensions of classic \ATs include the sequential-\AND gate (\SAND), indicating that subgoals must succeed in order from left to right \cite{JKM+15,AGKS15}.
\ATs without \SAND gates are called \emph{static}; those with \SAND{s} are called \emph{dynamic}.
A formal approach requires different semantics to these two categories, as we explain below.

\paragraph{Tree vs.\ \DAG attack trees}
Despite their name, \ATs are directed acyclic graphs (\DAGs) rather than trees, since subtrees can be shared by several parent nodes---see \Cref{fig:AT:example:dynamic}.
As elaborated below, \DAG-structured \ATs are computationally more challenging than those with a proper tree structure.

\paragraph{\AT metrics}
Besides learning the essential components and structure that constitute a feasible attack scenario, a vast number of algorithms have been developed to compute a wide range of \emph{security metrics}.
These metrics comprise key performance indicators (\acronym{kpi}{s}) that quantify relevant security features, such as the time, cost, and likelihood of different attack scenarios.
\acronym{kpi}{s} serve several purposes, e.g.\ allowing to compare different design alternatives w.r.t.\ the desired security features; compute the effectiveness of defensive measures; verify whether a solution meets its security requirements; etc.

\paragraph{\AT analysis}
%
Numerous algorithms have been proposed to compute security metrics. 
These include methods to compute the cost and probability of an attack \cite{BLP+06,JW08}, the time it takes \cite{KRS15,AHPS14,KSR+18}, as well as Pareto analyses that study trade offs between different attributes \cite{KRS15,FW19}.
Such algorithms exploit a wide plethora of techniques, for instance Petri nets \cite{DMCR06}, model checking \cite{AGKS15}, and Bayesian networks \cite{GIM15}.
While these algorithms provide good ways to compute metrics, they also suffer from several drawbacks:
\begin{enumerate*}[label=(\arabic*)]
\item	Many of them are geared to specific attributes, such as attack time
		or probability, while the procedure could extend to other metrics;
\item	Several algorithms do not exploit the acyclic structure of the \AT,
		specially approaches based on model checking;
\item	Since their application is mostly illustrated on small examples,
		it is unclear how these approaches scale to larger case studies.
\end{enumerate*}

\begin{figure}
	\hspace*{.7em}%
	\begin{minipage}[b]{.6\linewidth}
		\includegraphics[width=\linewidth]{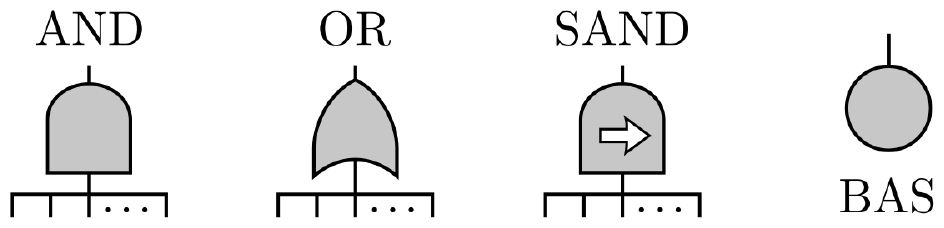}
	\end{minipage}
	\hfill
	\begin{minipage}[b]{.28\linewidth}
		\vfill~\\
		\captionof{figure}{\raggedright Nodes in an attack tree.}
		\label{fig:AT_nodes}
	\end{minipage}
	\vspace{-3ex}
\end{figure}

\hlbox{%
\paragraph{Approach}
We provide efficient and generic algorithms to compute \AT metrics,
by tailoring them to 
our 2-dimensional categorisation: 
static vs.\ dynamic \ATs, and proper trees vs.\ \DAG-structured \ATs.
These algorithms demand different semantics (that we provide) for dynamic attack trees.
}

Our algorithmic results are summarised in \Cref{tab:all_algos_intro}. An elaborate comparison with related work is provided in 
\Cref{sec:conclu}.


\begin{table*}
	\centering


\begingroup
\arrayrulecolor{gray}
\colorlet{rowcolA1}{white}
\colorlet{rowcolB1}{gray!28}
\colorlet{rowcolA2}{rowcolA1}
\colorlet{rowcolB2}{rowcolB1}
\rowcolors{2}{rowcolA1}{rowcolB1}
\extrarowheight=\aboverulesep
\addtolength{\extrarowheight}{\belowrulesep}
\aboverulesep=0pt
\belowrulesep=0pt
\def\metric#1{\parbox[c][5ex]{7.5em}{\bfseries #1}}
\def\algo#1#2{\parbox[c]{#1}{\centering{#2}}}
\def\OPEN{\textsf{\scshape\bfseries open~problem}}
\def\BU{\acronym{bu}}
\def\BUKordy{$\pazocal{C}$-\BU\xspace}
\def\head#1{\larger[.5]\sffamily\color{white}{#1}}
\begin{tabular}{>{~}l|>{}c|>{}c|>{}c>{}c|>{}c}
	\rowcolor{black!75}
		  \head{Metric}
		& \head{Static tree}
		& \head{Dynamic tree}
		& \multicolumn{2}{>{}c|}{\head{Static \DAG}}
		& \head{Dynamic \DAG}
	\\[.5ex]\midrule
	\metric{min cost}
		& \algo{8em}{\BU \cite{MO06,Wei91,Sch99}}
		& \algo{6em}{\BU \cite{JKM+15}}
		& \algo{7em}{\acronym{mtbdd} \cite{BET13}}
		& \algo{5em}{\BUKordy \cite{KW18}}
		& \algo{5em}{\acronym{pta} \cite{KRS15}}
	\\
	\metric{min time}
		& \algo{7em}{~\BU \cite{MO06,HAF+09}}
		& \algo{10em}{\acronym{aph} \cite{AHPS14} ~~ \BU \cite{JKM+15}}
		& \multicolumn{2}{>{}c|}{%
		  \algo{7em}{Petri nets \cite{DMCR06}}}
		& \algo{5em}{\acronym{pta} \cite{KRS15}}
	\\
	\metric{min skill}
		& \algo{7em}{~\BU \cite{MO06,BFM04}}
		& \algo{6em}{\BU \cite{JKM+15}}
		& \multicolumn{2}{>{}c|}{%
		  \algo{8em}{\BUKordy \cite{KW18}}}
		& ---
	\\
	\metric{max damage}
		& \algo{8em}{~\BU \cite{MO06,HAF+09,BFM04}}
		& \algo{6em}{\BU \cite{JKM+15}}
		& \algo{7em}{\acronym{mtbdd} \cite{BET13}}
		& \algo{5em}{\acronym{dpll} \cite{JW08}}
		& \algo{5em}{\acronym{pta} \cite{KRS15}}
	\\
	\metric{probability}
		& \algo{7em}{\BU \cite{BLP+06,HAF+09}}
		& \algo{7em}{\acronym{aph} \cite{AHPS14}}
		& \algo{7em}{\BDD \cite{Rau93}}
		& \algo{5em}{\acronym{dpll} \cite{JW08}}
		& \algo{7em}{\acronym{i/o-imc} \cite{AGKS15}}
	\\
	\metric{Pareto fronts}
		& \algo{7em}{\BU \cite{AN15,HAF+09}}
		& \algo{7.1em}{\OPEN}
		& \multicolumn{2}{>{}c|}{%
		  \algo{9em}{\BUKordy \cite{FW19}}}
		& \algo{6em}{\acronym{pta} \cite{KRS15}}
	\\
	\rowcolor{rowcolB2}
	\metric{Any of the above}
		& \algo{7em}{\bfseries\Cref{alg:bottom_up_SAT}:~$\mathtt{BU_{SAT}}$}
		& \algo{7em}{\bfseries\Cref{alg:bottom_up_DAT}:~$\mathtt{BU_{DAT}}$}
		& \multicolumn{2}{>{}c|}{%
		  \algo{9em}{\bfseries\Cref{alg:bottom_up_BDD}:~$\mathtt{BDD_{DAG}}$}}
		& \algo{7.1em}{\OPEN}
	\\
	\rowcolor{rowcolA2}
	\metric{$\boldsymbol{k}$-top metrics}
		& \algo{8em}{\BU-projection \cite{MO06}}
		& \algo{7.1em}{\OPEN}
		& \multicolumn{2}{>{}c|}{%
		  \algo{13.5em}{\bfseries\Cref{alg:shortest_path_BDD}: $\mathtt{BDD~shortest\_paths}$}}
		& \algo{7.1em}{\OPEN}
	\tikz[overlay]{\draw[black,very thick] (-15.1,-.25) rectangle ++(15.3,1.28);}
	\\
\end{tabular}
\endgroup

	\vspace{1ex}
	\caption{Efficient algorithms to compute security metrics on different
	         AT classes~~\textcolor{black!75}{(details in \Cref{sec:conclu})}}
	\label{tab:all_algos_intro}
	\vspace{-2ex}
\end{table*}

\paragraph{Static trees}
We start with the simplest category: static attack trees (\SATs) with proper tree structure.
As shown in a seminal paper by Mauw \& Oosdijk \cite{MO06}, metrics can be computed for tree-structured \SATs in a bottom-up fashion.
This algorithm propagates values from the leaves to the top, using appropriate operators $\operOR$ and $\operAND$ resp.\ for the \OR and \AND gates in the tree.
We show this as \Cref{alg:bottom_up_SAT}.
A key insight in \cite{MO06} is that this procedure works whenever the algebraic structure $(\Vdom,\operAND,\operOR)$ constitutes a semiring.
In particular, $\operAND$ must distribute over $\operOR$.
\\[.3ex]
We provide an alternative proof of correctness for this result: while \cite{MO06} deploys rewriting rules for attack trees, we work directly on the syntactic \AT structure.
Furthermore, we propose new classes of attribute domains, which extend the application of the bottom-up algorithm to compute popular security metrics, including stochastic and Pareto analyses.

\paragraph{Static \DAGs}
It is well-known that static attack trees with \DAG structure cannot be analysed via a bottom-up procedure \cite{Rau93,BK18}.
Several algorithms have been devised to tackle with such \ATs, mostly geared to specific metrics \cite{AGKS15,BET13,DMCR06,JW08,KW18}.
\\[.3ex]
\emph{A key contribution of this paper is a generic algorithm (\Cref{alg:bottom_up_BDD}) that works over any semiring attribute domain $(\Vdom,\operOR,\operAND)$, i.e.\ where $\operAND$ distributes over $\operOR$}.

Concretely, we exploit a binary decision diagram representation (\BDD) of the attack tree.
Our algorithm visits each \BDD node once and is thus linear in its size.
The caveat is that \BDDs can be of exponential size in the number of \BAS, but one cannot hope for faster algorithms: as we show, computing a minimal attack is an NP-hard problem.
Moreover, \BDDs are known to be compact in practice \cite{Bry86}, and allow parallel traversals \cite{ODP17}, making them an overall efficient choice.

\paragraph{Dynamic trees}
A challenge to compute metrics for dynamic attack trees (\DATs) is to define them formally based on their semantics.
Usually metrics are decoupled from semantics, and defined either on the syntactic \AT structure, or ad hoc for the selected computation method \cite{JKM+15,KRS15,ANP16,KW18}.
A main obstacle is to choose semantics for \DATs in a way that supports a proper definition of metric, i.e.\ that is compatible with the notion of metric of static \ATs, and that is generic as the attribute domains from \cite{MO06}.
In particular, the interaction among multiple \SAND gates is nontrivial, because they may impose conflicting execution orders on the \BAS of the tree.
\emph{One of our key contributions is to define a notion of well-formedness that rules out conflicting requirements}.

We give semantics to well-formed \DATs in terms of partially ordered sets (\emph{posets}).
Each poset \poset represents an attack scenario, where \attack collects all attacks steps to be performed, and $a\prec b$ indicates that step $a$ must be completed before step $b$ starts.
%
%
%
This set up enables us to define a notion of metric for \DATs based on their semantics.
\emph{We then show that tree-structured DATs are analysable by extending the bottom-up algorithm with an additional operator (see \Cref{alg:bottom_up_DAT})}.
Concretely, we use attribute domains with three operators: $\operOR$, $\operAND$, $\operSAND$, where $\operSAND$ distributes over $\operOR$ and $\operAND$, and $\operAND$ over $\operOR$.
We prove this \namecref{alg:bottom_up_DAT} correct in our formal semantics.
Note that earlier algorithms do not provide explicit correctness results in terms of semantics. 
Our result is non-trivial, because the metrics are formally defined on the (poset) semantics of a \DAT, while the algorithm works on its syntactic AT structure. 

\paragraph{Dynamic \DAGs}
Efficient computation of metrics for \DAG-structured \DATs is left as future research challenge.
A na\"ive, inefficient algorithm would enumerate all posets in the semantics.
Instead, one could extend \BDD-algorithms for static \DAGs to dynamic \ATs.
This is non-trivial: \BDDs ignore the order of attack steps.
Thus, efficient analysis of \DAG-structured dynamic \ATs is an important open problem.


\hlbox{%
\paragraph{Contributions}
In summary, our contributions are: 
\begin{enumerate}[label=\textbf{\arabic*.},leftmargin=1.3em]
\item	An efficient and generic \BDD-based algorithm for \DAG-\SATs,
		working for semiring attribute domains $(\Vdom,\operOR,\operAND)$;
\item	A theorem proving that computing a minimal successful attack is NP-hard;
\item	An algorithm to compute the $k$-top best attacks;
\item	A novel and intuitive poset semantics for dynamic attack trees that better matches the order behavior of \SAND{s};
\item	A bottom-up algorithm for tree-structured \DATs; 
\item	Future directions to analyse \DAG-\DATs efficiently
		(identified as an open problem).
\end{enumerate}
We place ourselves in the literature in \Cref{tab:all_algos_intro} and \Cref{sec:conclu}.
}

\paragraph{Paper structure}
We introduce all essential concepts and our formal syntax of attack trees in \Cref{sec:AT}\textbf{.}\;
\Crefrange{sec:SAT}{sec:SAT_DAGs} study static \ATs, and \Crefrange{sec:DAT}{sec:DAT_DAGs} study dynamic \ATs.
The paper concludes in \Cref{sec:conclu}, revising related work.


\section{Attack Trees}
\label{sec:AT}

\subsection{Attack tree models}
\label{sec:AT:models}

Syntactically, an attack tree is a rooted \DAG that models an undesired event caused by a malicious party, e.g.\ a security breach.
\ATs show a top-down decomposition of a top-level attack---the unique root of the \DAG---into simpler steps.
The leaves are basic steps carried out by the attacker.
The nodes between the basic steps and the root are intermediate attacks, and are labelled with gates to indicate how its input nodes (children) combine to make the intermediate attack succeed.

\paragraph{Basic Attack Steps}
The leaves of the \AT represent indivisible actions carried out by the attacker, e.g.\ smash a window,
decrypt a file by brute-force attack,
etc.
These \BAS nodes can be enriched with attributes, such as its execution time, the cost incurred, and the probability with which the \BAS occurs.
We model attributes via an attribution function $\attrOp\from\BAS\to\Vdom$.

\paragraph{Gates}
Non-leaf nodes serve to model intermediate attacks that lead to the \emph{top-level attack} (\TLA).
Each has a logical \emph{gate} that describes how its children combine to make it succeed:
an \OR gate means that the intermediate attack will succeed if any of its child nodes succeeds;
an \AND gate indicates that all children must succeed, in any order or possibly in parallel;
a \SAND gate (viz.\ sequential-\AND) needs all children to succeed sequentially in a left-to-right order.



\begin{figure}
	\centering
	\def\HEIGHT{23ex}
	\begin{subfigure}[b]{.48\linewidth}
		\centering\hspace*{-2em}%
		\includegraphics[height=\HEIGHT]{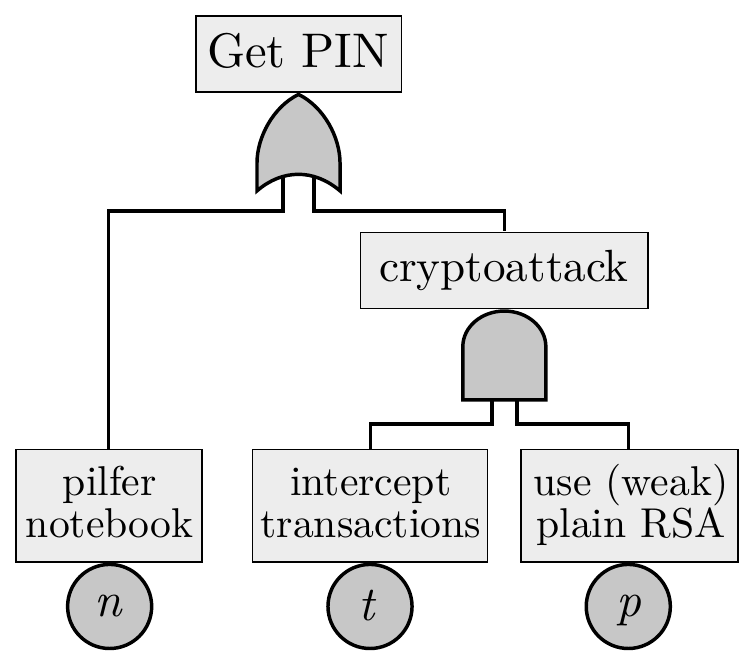}
		\caption{A static-tree \AT: \sampleTs\hspace*{2em}}
		\label{fig:AT:example:static}
	\end{subfigure}
	\hspace{-1em}
	\begin{subfigure}[b]{.48\linewidth}
		\centering
		\includegraphics[height=\HEIGHT]{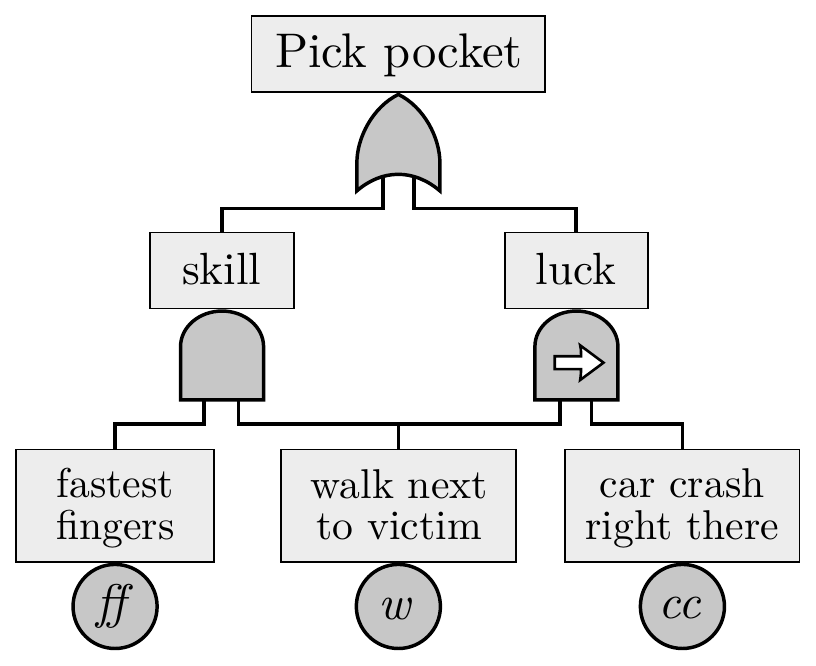}
		\caption{A dynamic-\DAG \AT: \sampleTd}
		\label{fig:AT:example:dynamic}
	\end{subfigure}
	\caption{Attack tree models}
	\label{fig:AT}
	\vspace{-2ex}
\end{figure}

\begin{example}
	\label{ex:running_examples}
	\Cref{fig:AT:example:static} shows a static attack tree, \sampleTs, that models how a \acronym{pin} code can be obtained by either pilfering a notebook, or via a cryptographic attack.
	The pilfering is considered atomic, while the cryptoattack consists of two steps which must both succeed: intercepting transactions, and abusing weak \acronym{rsa} encryption. 
	Note that \sampleTs has a plain tree structure.
	Instead, \Cref{fig:AT:example:dynamic} shows a dynamic attack tree, \sampleTd, with a \DAG structure.
	Its \TLA is to pick a pocket, which is achieved either by having ``skill'' or ``luck.''
	In both cases the attacker must walk next to the victim, so these gates share the \BAS child $w$, making \sampleTd not a tree.
	In the case of ``luck'' the order of events matters: if the attacker first walks next to the victim and then a traffic accident happens, the pick-pocket succeeds.
	Thus, this intermediate attack is modelled with a \SAND gate.
	Instead, ``fastest fingers'' is an inherent attacker flair that is always present.
	It is thus meaningless to speak of an order w.r.t.\ the attacker-victim encounter, so an \AND gate is used.
\end{example}

\paragraph{Security metrics}
A key goal in quantitative security analysis is to compute relevant \emph{security metrics}, which quantify how well a system performs in terms of security.
Typical examples are the cost of the cheapest attack, the probability of the most likely one, the damage produced by the most harmful one, and combinations thereof.
Security metrics for \ATs are typically obtained by combining the attribute valuations $\attr{a}\in\Vdom$ assigned to the \BAS.
For example, the cheapest attack is the attack where the sum of the cost of the \BAS is minimal.
\emph{The key topic of this paper is to compute a large class of security metrics in a generic and efficient way}.
For this we give different (formal) semantics to static and dynamic \ATs, and introduce algorithms based on these semantics and the \AT structure.
We begin by formalising the notion of~\AT~model.

\subsection{Attack tree syntax}
\label{sec:AT:syntax}

\ATs are rooted \DAG{s} with typed nodes: we consider types $\mathbb{T}=\{\tBAS,\tOR,\tAND,\tSAND\}$.
For Booleans we use $\BB=\{\top,\bot\}$.
The edges of an \AT are given by a function $\chOp$ that assigns to each node its (possibly empty) sequence of children.
We use set notation for sequences, e.g.\ $e\in(e_1,\ldots,e_m)$ means $\exists i.\,e_i=e$, and we denote the empty sequence by $\varepsilon$.

\begin{definition}
	\label{def:AT:syntax}
	An \emph{attack tree} is a tuple
	\mbox{$\T=\left(\ATnodes,\typOp,\chOp\right)$} where:
	\begin{itemize}[topsep=.5ex,parsep=.1ex,itemsep=0pt]
	\item	$\ATnodes$ is a finite set of \emph{nodes};
	\item	$\typOp\from \ATnodes\to\mathbb{T}$
			gives the \emph{type} of each node;
	\item	$\chOp\from \ATnodes\to\ATnodes^\ast$
			gives the sequence of \emph{children} of a node.
	\end{itemize}
	Moreover, \T satisfies the following constraints:
	\begin{itemize}[topsep=.5ex,parsep=.1ex,itemsep=0pt]
	\item	$(\ATnodes,\mathobject{E})$ is a connected \DAG, where\\
			\phantom{.}\hfill%
			$\mathobject{E}=\left\{(v,u)\in\ATnodes^2\mid u\in\child{v}\right\}$;
	\item	\T has a unique root, denoted \ATroot:\\
			\phantom{.}\hfill%
			$\exists!\,\ATroot\in \ATnodes.~%
			\forall v\in \ATnodes.~\ATroot\not\in\child{v}$;
	\item	{$\normalfont\BAS_{\T}$} nodes are the leaves of \T:\\
			\phantom{.}\hfill%
			$\forall v\in \ATnodes.~%
			{\type{v}=\tBAS} \Leftrightarrow {\child{v}=\varepsilon}$.
	\end{itemize}
\end{definition} 

We omit the subindex \T if no ambiguity arises, e.g.\ an attack tree $\T=(\ATnodes,\typOp,\chOp)$ defines a set $\BAS\subseteq\ATnodes$ of basic attack steps.
If \mbox{$u\in\child{v}$} then $u$ is called a \emph{child} of $v$, and $v$ is a \emph{parent} of $u$.
Moreover we write $v=\AND(v_1,\ldots,v_n)$ if ${\type{v}=\tAND}$ and ${\child{v}=(v_1,\ldots,v_n)}$, and analogously for $\OR$ and $\SAND$.
We denote the universe of \ATs by \allATs and call $\T\in\allATs$ \emph{tree-structured} if
${\forall v,u\in\ATnodes. \child{v}\cap\child{u}=\varepsilon}$;
else we say that \T has \emph{DAG structure}.



\section{Analysis of Static Attack Trees}
\label{sec:SAT}

In the absence of \SAND gates the order of execution of the \BAS is irrelevant.
This allows for simple semantics given in terms of a Boolean function called \emph{structure function}.
The computation of security metrics, however, crucially depends on whether the \AT structure is a tree or a \DAG.

\subsection{Semantics for static attack trees}
\label{sec:SAT:semantics}

The semantics of a static attack tree (\SAT) is defined by its successful attack scenarios, in turn given by its structure function.
First, we define the notions of attack and attack suite.

\begin{definition}
	\label{def:SAT:attack}
	An \emph{attack scenario}, or shortly an \emph{attack},
	of a static \AT \T is a subset of its basic attack steps:
	\mbox{$\attack\subseteq\BAS_{\T}$}.
	An \emph{attack suite} is a set of attacks
	\mbox{$\suite\subseteq2^{\BAS_{\T}}$}.
	We denote by 
	$\allAttacks_{\T}={2^{\BAS_{\T}}}$
	the universe of attacks of \T, and by 
	$\allSuites_{\T}=2^{2^\BAS}$
	the universe of attack suites of \T.
\end{definition}

Intuitively, an attack suite $\suite\in\allSuites$ represents different ways in which the system can be compromised.
From those, one is interested in attacks $\attack\in\suite$ that actually represent a threat.
For instance for \sampleTs in \Cref{ex:running_examples} one such attack is $\{t,p\}$.
In contrast, $\{t\}$ is an attack that does not succeed, i.e.\ it cannot cause a \TLA.
The structure function $\sfunT(v,\attack)$ indicates whether the attack $\attack\in\allAttacks$ succeeds at node $v\in\ATnodes$ of \T.

\begin{definition}
	\label{def:SAT:sfun}
	The \emph{structure function} $\sfunT\from\ATnodes\times\allAttacks\to\BB$
	of a static attack tree \T is given by:
	\begin{align*}
	  \sfunT(v,\attack) =&
	  \begin{cases}
		\top  & \parbox{55pt}{if~$\type{v}=\tOR$}~\text{and}~%
				\exists u\in\child{v}.\sfunT(u,\attack)=\top,\\
		\top  & \parbox{55pt}{if~$\type{v}=\tAND$}~\text{and}~%
				\forall u\in\child{v}.\sfunT(u,\attack)=\top,\\
		\top  & \parbox{55pt}{if~$\type{v}=\tBAS$}~\text{and}~%
				v\in\attack,\\
		\bot  & \text{otherwise}.
	  \end{cases}
	\end{align*}
\end{definition}

We let $\sfunT(\attack) \doteq \sfunT(\ATroot,\attack)$.
An attack \attack is called \emph{successful} if $\sfunT(\attack)=\top$, i.e.\ it makes the \TLA of \T succeed; if moreover no proper subset of \attack is successful then \attack is a \emph{minimal attack}.

\SATs are \emph{coherent} \cite{BP75}, meaning that adding attack steps preserves success: if \attack is successful then so is $\attack\cup\{a\}$ for any $a\in\BAS$.
Thus, the suite of successful attacks of an \AT is characterised by its minimal attacks.
This was first formalised in \cite{MO06}, and is called multiset semantics in \cite{WAFP19}:
\begin{definition}
	\label{def:SAT:semantics}
	The \emph{semantics of a static \AT} \T is its suite of minimal attacks:
	$\amin = \left\{ \attack\in\allAttacks_{\T} \mid
		\sfunT(\attack)\land \attack~\text{is minimal} \right\}$.
\end{definition}

\begin{example}
	\label{ex:SAT:semantics}
	The static \AT in \Cref{ex:running_examples}, \sampleTs
	(\Cref{fig:AT:example:static}), has three successful attacks:
	$\{n\}$, $\{t,p\}$, and $\{n,t,p\}$. The first two are minimal,
	so we have: $\ssem{\sampleTs}=\{\{n\},\{t,p\}\}$.
\end{example}

An alternative characterisation of this semantics for tree-structured \SATs is shown as \Cref{lemma:ssem}, which also provides the key argument for correctness of the bottom-up procedure (\Cref{alg:bottom_up_SAT} in \Cref{sec:SAT_trees}). 
\Cref{lemma:ssem} can be used to compute the semantics of \Cref{def:SAT:semantics} by recursively applying cases \textit{\ref{lemma:ssem:BAS})--\ref{lemma:ssem:AND})} to $\ssem{\ATroot}\doteq\ssem{\T}$.
However, \BDD representations provide more compact encodings of this semantics (see \Cref{sec:SAT_DAGs}).

We formulate \Cref{lemma:ssem} for binary \ATs; its extension to arbitrary trees is straightforward but notationally cumbersome.
\ifthenelse{\boolean{forCSF}}%
{
The proof uses induction on the structure of an \AT, whose root is the node $v$ in each left-hand side \sem{v} of cases~\textit{\ref{lemma:ssem:BAS})--\ref{lemma:ssem:AND})}.
Case~\textit{\ref{lemma:ssem:disjoint})} is trivial by the tree-structure.
We give the full proof in \cite{arXiv}.
}{
The proof of \Cref{lemma:ssem} is given in \Cref{sec:proofs}, \cpageref{lemma:ssem:proof}.
}

\begin{lemma}
	\label{lemma:ssem}
	Consider a \SAT with nodes $a\in\BAS, v_1, v_2\in\ATnodes$,
	that has a proper tree structure. Then:
	\begin{enumerate}
	\def\REF#1{\textit{\ref{#1})}}
	\item	$\ssem{a} = \{ \{a\}\}$;
			\label{lemma:ssem:BAS}
	\item	$\ssem{\OR(v_1,v_2)} = \ssem{v_1} \cup \ssem{v_2}$;
			\label{lemma:ssem:OR}
	\item	$\ssem{\AND(v_1,v_2)} = \{ \attack_1\cup \attack_2 \mid
				\attack_1\in\ssem{v_1} \land \attack_2\in\ssem{v_2} \}$;
			\label{lemma:ssem:AND}
	\item	In cases \REF{lemma:ssem:OR} and \REF{lemma:ssem:AND}
			the \ssem{v_i} are disjoint, and in case \REF{lemma:ssem:AND}
			moreover the $A_i$ are pairwise disjoint.
			\label{lemma:ssem:disjoint}
	\end{enumerate}
\end{lemma}

\subsection{Security metrics for static attack trees}
\label{sec:SAT:metrics}

\Cref{lemma:ssem} allows for \emph{qualitative analyses}, i.e.\ finding the minimal sets of \BAS that lead to a \TLA.
To enable \emph{quantitative analyses}, i.e.\ computing security metrics such as the minimal time and cost among all attacks, all \BAS are enriched with attributes.
We thus define security metrics in three steps:
first an attribution $\attrOp$ assigns a value to each \BAS;
then a security metric $\metrAOp$ assigns a value to each attack scenario;
and finally the metric $\metrOp$ assigns a value to each attack suite.


\begin{definition}
	\label{def:metric}
	Given an \AT and a set \Vdom of values:
	\begin{enumerate}
	\item	an \emph{attribution} $\attrOp\from\BAS\to\Vdom$
			assigns an \emph{attribute value} \attr{a},
			or shortly an \emph{attribute}, to each basic attack step $a$;
	\item	a \emph{security metric} refers both to a function
			$\metrAOp\from\allAttacks_{\T}\to\Vdom$ that assigns a value
			$\metrA{\attack}$ to each attack $\attack$;
			\newline
			and to a function $\metrOp\from\allSuites_{\T}\to\Vdom$ that
			assigns a value $\metr{\suite}$ to each attack suite $\suite$.
	\end{enumerate}
	We write $\metr{\T}$ for $\metr{\ssem{\T}}$, setting the metric of an \AT to the metric of its minimal attack suites. 
\end{definition}

\begin{example}
	\label{ex:metric}
	Let $\Vdom=\NN$ denote time, so that \attr{a} gives the time required
	to perform the basic attack step $a$.
	Then the time needed to complete an attack \attack can be given by
	$\metrA{\attack} = \sum_{a\in\attack} \attr{a}$,
	and the time of the fastest attack in a suite \suite is
	$\metr{\suite} = \min_{\attack\in\suite} \metrA{\attack}$.
	If instead $\Vdom=[0,1]\subset\RR$ denotes probability,
	then the probability of an attack is given by
	$\metrA{\attack} = \prod_{a\in\attack} \attr{a}$,
	and the probability of the likeliest attack in a suite is
	$\metr{\suite} = \max_{\attack\in\suite} \metrA{\attack}$.
\end{example}

\Cref{def:metric} gives a lax notion of metric.
For a more concise definition---that enables computation for static \ATs, but does not depend on their tree/\DAG-structure---one must resort to the semantics.
For this we follow an approach similar to that of Mauw and Oostdijk \cite{MO06}.
Namely, we define a \emph{metric function} $\metrOp\from \allATs \to\Vdom$ that yields a value for each \SAT based on its semantics, an attribution, and two binary operators $\operOR$ and $\operAND$.

\begin{definition}
	\label{def:SAT:metric}
	\def\VH{\vphantom{\bigoperAND_{\ssem{\T}}}}
	Let \Vdom be a set:
	\begin{enumerate}
	\item an \emph{attribute domain} over \Vdom is a tuple $\domain=(\Vdom,\operOR,\operAND)$,
	whose \emph{disjunctive operator} \mbox{$\operOR\from\Vdom^2\to\Vdom$},
	and \emph{conjunctive operator} \mbox{$\operAND\from\Vdom^2\to\Vdom$},
	are associative and commutative;
	\item the attribute domain is a \emph{semiring}\footnote{Since we require $\operAND$ to be commutative, $\domain$ is in fact a commutative semiring. Further, rings often include a neutral element for disjunction and an absorbing element for conjunction, but these are not needed in \Cref{def:SAT:metric}.}
	if $\operAND$ distributes over $\operOR$, i.e.
	$\forall\,x,y,z\in\Vdom.\;%
		x\operAND(y\operOR z)=(x\operAND y)\operOR(x\operAND z)$;
	\item
	let \T be a static \AT and $\attrOp$ an attribution on \Vdom.
	The \emph{metric for \T} associated to $\attrOp$ and $\domain$ is given by:
	\begin{align*}
		\metr{\T} &~=
			\underbrace{\VH\bigoperOR_{\attack\in\ssem{\T}}\:}_{\mathlarger\metrOp}\,
			\underbrace{\VH~\bigoperAND_{a\in\attack}~}_{\mathlarger\metrAOp}
			\attr{a}.
	\end{align*}
	\end{enumerate}
\end{definition}

\begin{example}
	\label{ex:SAT:metric}
	Consider the static \AT $\sampleTs=\OR\big(n,\AND(t,p)\big)$ from \Cref{fig:AT:example:static}, and recall that $\ssem{\sampleTs}=\{\{n\},\{t,p\}\}$.
	Let ${\Vdom=\NN}$ denote time as in \Cref{ex:metric}, and consider an attribution ${\attrOp}=\{{n\mapsto1}, {t\mapsto100}, {p\mapsto0}\}$.
	Then the metric for the fastest attack time is given by the attribute domain $(\Vdom,\minb,+)$: 
	\begin{align*}
	\metr{\sampleTs}
		&= \bigoperOR_{\attack\in\{\{n\},\{t,p\}\}}~%
		   \bigoperAND_{a\in\attack} \attr{a}\\
		&=~ \attr{n}\operOR\big(\attr{t}\operAND\attr{p}\big)
		~=~ 1 \minb (100 + 0)
		~=~ 1,
	\end{align*}
	where $\minb$ has infix notation, i.e.\ ${x \minb y = \min(x,y)}$.
	For probability, let ${\Vdom'=[0,1]}$ and ${\attrOp'}=\{{n\mapsto0.07},{p\mapsto0.01},$ $ {t\mapsto0.95}\}$.
	Then the attribute domain $(\Vdom',\maxb,\ast)$ allows to compute the probability of the likeliest attack:
	$\metrOp'(\sampleTs)
		= {\attrOp'(n) \operOR' \big(\attrOp'(t) \operAND' \attrOp'(p)\big)}
		= {0.07 \maxb (0.95 \ast 0.01)} ~=~ 0.07.$
\end{example}


\section{Computations for tree-structured SATs}
\label{sec:SAT_trees}

\Cref{ex:SAT:metric} illustrates how to compute metrics for static \ATs using \Cref{def:SAT:metric}.
However, this method requires to first compute the semantics of the attack tree, which is \emph{exponential} in the number of nodes \card{\ATnodes}---see \Cref{theo:NP_hard} in \Cref{sec:SAT_DAGs}, or \cite{KW18}.

A key result in  \cite{MO06} is that metrics defined on attribute domains $(\Vdom,\operOR,\operAND)$ that are semirings, can be computed via a bottom-up algorithm that is \emph{linear} in \card{\ATnodes} as long as the static \AT has a proper tree structure.
We repeat this result here, giving a more direct proof of correctness, and extending it to dynamic attack trees in \Cref{sec:SAT_DAGs}.

\subsection{Bottom-up algorithm}
\label{sec:SAT_trees:algorithm}

First we formulate the procedure as \Cref{alg:bottom_up_SAT}, which propagates the attribute values from the leaves of the \SAT to its root, interpreting \OR gates as $\operOR$ and \AND{s} as $\operAND$.
This algorithm is clearly linear in \card{\ATnodes} since each node in the tree \T is visited once.
\Cref{alg:bottom_up_SAT} can be called on any node of \T: to compute the metric \metr{\T} it must be called on its root node \ATroot.


\begin{algorithm}
	\KwIn{Static attack tree $\T=(\ATnodes,\typOp,\chOp)$,\newline
	      node $v\in\ATnodes$,\newline
	      attribution $\attrOp$,\newline
	      semiring attribute domain $\domain=(\Vdom,\operOR,\operAND)$.}
	\KwOut{Metric value $\metr{\T}\in\Vdom$.}
	\BlankLine
	\uIf{$\type{v}=\tOR$}{%
		\Return{$\bigoperOR_{u\in\child{v}}
		         \BUSAT(\T,u,\attrOp,\domain)$}
	} \uElseIf{$\type{v}=\tAND$}{%
		\Return{$\bigoperAND_{u\in\child{v}}
		         \BUSAT(\T,u,\attrOp,\domain)$}
	} \Else(\tcp*[h]{$\type{v}=\tBAS$}) {%
		\Return{\attr{v}}
	}
	\caption{\BUSAT for a tree-structured \SAT \T}
	\label{alg:bottom_up_SAT}
\end{algorithm}

We state the correctness of \Cref{alg:bottom_up_SAT} in \Cref{theo:bottom_up_SAT},
\ifthenelse{\boolean{forCSF}}%
{
whose proof relies on \Cref{lemma:ssem}, and therefore works by structural induction on the (binary) tree \T---we provide the full proof in \cite{arXiv}.
}{
which we prove in \Cref{theo:bottom_up_SAT:proof}, \cpageref{theo:bottom_up_SAT:proof}.
}
This result was proven in \cite{MO06} via rewriting rules for \ATs with a slightly different structure denoted ``bundles.''
Our result concerns attack trees in the syntax from \Cref{def:AT:syntax}, which is more conforming to the broad literature \cite{Wei91,Sch99,BLP+06,JW08,KMRS11,BET13,KRS15}.

\begin{theorem}
	\label{theo:bottom_up_SAT}
	Let \T be a static \AT with tree structure,
	$\attrOp$ an attribution on \Vdom,
	and $\domain=(\Vdom,\operOR,\operAND)$ a semiring attribute domain.
	Then $\metr{\T} = \BUSAT(\T,\ATroot,\attrOp,\domain)$.
\end{theorem}
\subsection{Metrics as semiring attribute domains}
\label{sec:SAT_trees:metrics}

Many relevant metrics for security analyses on \SATs can be formulated as semiring attribute domains.
\Cref{tab:SAT:metric} shows examples, where ${[0,1]_\QQ=[0,1]\cap\QQ}$, and $\NN_\infty=\NN\cup\{\infty\}$ includes $0$ and $\infty$.
For instance ``min cost'' can be formulated in terms of ${(\NN_\infty,\min,+)}$, which is a semiring attribute domain because $+$ distributes over $\min$, i.e.\ ${a+(b\minb c)}={(a+b)\minb(a+c)}$ for all ${a,b,c\in\NN_\infty}$.
Also, attribute domains can handle \SAND gates providing that the execution order is irrelevant for the metric.
This works for example with min skill and max damage.

\paragraph{Non-semiring metrics}
Nevertheless, some meaningful metrics do fall outside this category. 
For instance and as observed in~\cite{MO06}, the cost to defend against all attacks is represented by $(\NN_\infty,+,\min)$, but since $\min$ does not distribute over $+$ (i.e.\ in general ${a\minb(b+ c)}\neq{(a\minb b)+(a\minb c)}$) then this metric cannot be computed via \Cref{alg:bottom_up_SAT}.
Less well-known is that the total attack probability---given by $\metr{\T}=\sum_{\attack\in\ssem{\T}} \metrA{\attack}$ where $\metrA{\attack}=\big(\prod_{a\in \attack} \attr{a}\big) \cdot \big(\prod_{a\not\in \attack} (1-\attr{a})\big)$---can neither be formulated as an attribute domain.
The problem is that $\metrA{\attack}$ does not have the shape $\bigoperAND_{a\in\attack} \attr{a}$.
Interestingly though, this probability can still be computed via a bottom-up procedure by taking $\metr{\AND(v_1,v_2)} = \metr{\ssem{v_1}} \ast \metr{\ssem{v_2}}$ and $\metr{\OR(v_1,v_2)} = \metr{\ssem{v_1}} + \metr{\ssem{v_2}} - \metr{\ssem{v_1}\cap\ssem{v_2}}$.


\begin{table}
  \centering
  \begin{tabular}{l>{}c>{\quad}c>{\quad}c}
	\toprule
	\scshape Metric       & $\Vdom$           & $\operOR$ & $\operAND$ \\
	\midrule
	min cost              & $\NN_\infty$      & $\min$ & $+$    \\
	min time (sequential) & $\NN_\infty$      & $\min$ & $+$    \\
	min time (parallel)   & $\NN_\infty$      & $\min$ & $\max$ \\
	min skill             & $\NN_\infty$      & $\min$ & $\max$ \\
	max challenge         & $\NN_\infty$      & $\max$ & $\max$ \\
	max damage            & $\NN_\infty$      & $\max$ & $+$    \\ 
	discrete prob.        & $[0,1]_\QQ$       & $\max$ & $\ast$ \\
	continuous prob.      & $\RR\to[0,1]_\QQ$ & $\max$ & $\ast$ \\
	\bottomrule
  \end{tabular}
  \caption{\SAT metrics with semiring attribute domains}
  \label{tab:SAT:metric}
\end{table}

\paragraph{Stochastic analyses}
Semirings are closed under finite and infinite products \cite{Mac71}: this allows to propagate not only tuples of attribute values, but also functions over them. 
In particular, \emph{cumulative density functions} that assign a probability $t\mapsto P[X\leqslant t]$ constitute a semiring \cite{AHPS14}.
Such functions are useful, e.g.\ to consider attack probabilities, cost, or damage, as functions that evolve on time.

\paragraph{Pareto analyses}
Moreover, \emph{Pareto frontiers} can be formulated as semirings.
Pareto analysis is a cornerstone in multi-parameter optimisation, that seeks the dominant (i.e.\ best-performing) solutions over multiple attributes.
A solution is called \emph{Pareto-efficient} if it is not dominated by any other solution in the ordering relation \cite{LLCM10}.
For example consider three attack scenarios:
$\attack_1$ that takes $2$ time units and has cost $3$;
$\attack_2$ with time $1$ and cost $3$; and
$\attack_3$ with time $2$ and cost $1$.
Then attack $\attack_1$ is not Pareto-efficient because $A_2$ is faster at same cost.
On the other hand, $\attack_2$ and $\attack_3$ are incomparable because the former is faster while the latter is cheaper.
So among these three attack scenarios, $\attack_2$ and $\attack_3$ are in the Pareto frontier.
Pareto frontiers are sets of Pareto-efficient solutions: for \AT metrics these are cross-products of semiring attribute domains, which preserve the semiring property \cite{Mac71}.



\section{Computations for DAG-structured SATs}
\label{sec:SAT_DAGs}

Attack trees with shared subtrees cannot be analysed via a bottom-up procedure on its (\DAG) structure, as we illustrate next in \Cref{ex:bottom_up_DAG}.
This is a classical result from fault tree analysis \cite{LGTL85}, later rediscovered for attack trees e.g.\ in \cite{KW18}.

Various methods to analyse \DAG-structured \ATs have been proposed: see \Cref{tab:all_algos_intro} for contributions over the last 15 years, including \cite{AGKS15,BET13,DMCR06,JW08,KW18}.
These methods are often geared towards specific metrics, e.g.\ cost, time, or probability \cite{BLP+06,JW08,AHPS14}.
Others use general-purpose techniques of high complexity and low efficiency, such as model checking \cite{DMCR06,KRS15}. 

We present a novel algorithm based on a binary decision diagram (\BDD) representation of the structure function of the static \AT.
\BDDs offer a very compact encoding of Boolean functions, and are heavily used in model checking \cite{HS99,BK08,KP12}, as well as for probabilistic fault tree analysis \cite{Rau93,RS15b}.

Our \BDD-based approach works for semiring attribute domains (with neutral elements for the operators $\operOR$ and $\operAND$) regardless of the \AT structure.
It thus extends the generic and efficient result of \cite{MO06}---that works for tree-structure \SATs only---to include \DAG-structured \SATs as well.

Our algorithm traverses the \BDD bottom-up, which makes it linear in its size.
\BDDs, however,  can be exponential in the tree size \cite{RD97}. Below, we show that the problem of computing metrics is NP-hard, so no asymptotically-faster algorithms can be found.
Moreover, \BDDs are among the most efficient approaches in terms of practical performance \cite{Bry86,BET13}.

\subsection{Computational complexity}
\label{sec:SAT_DAGs:complexity}

We first show why the bottom-up procedure fails to compute metrics for \ATs that have shared subtrees. 


\begin{figure}
  \vspace{-2ex}
  \centering
  \def\caja#1{\parbox[b]{.5em}{\centering$#1$}}

  \begin{subfigure}[b]{.27\linewidth}
	\includegraphics[width=\linewidth]{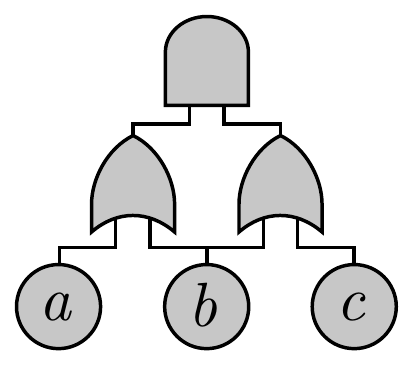}
  \end{subfigure}
  ~
  \begin{subfigure}[b]{.66\linewidth}
	\def\mp#1#2{\begin{minipage}[b]{#1\linewidth}{{#2}}\end{minipage}}
	\mp{.08}{%
		Let:\vspace{6.5ex}
	}
	\mp{.80}{%
	  \begin{align*}
		\attr{\caja{a}} &= 3  & \Vdom    &= \NN_\infty  \\[-.5ex]
		\attr{\caja{b}} &= 1  & \operOR  &= \min        \\[-.5ex]
		\attr{\caja{c}} &= 4  & \operAND &= {+}
	  \end{align*}
	}
	\mp{.08}{\hspace*{-2em}
	\tikz[baseline=-.4ex]{\draw[x=1ex,y=1ex,decorate,decoration={%
		                        brace,mirror,amplitude=5pt}] (0,0)--(0,8.1)
	                      node [midway,xshift=1.3em] {\domain};}}
	
	The cheapest attack is $\{b\}$: $\metrA{\{b\}}=1$.
	\vspace{0.3ex}
  \end{subfigure}
  
  \caption{Metrics cannot be computed bottom-up on \ATs with \DAG structure.
           For min cost in this static \AT, \Cref{alg:bottom_up_SAT} yields:
           $\BUSAT(\T,\ATroot,\attrOp,\domain) = (3\minb1) + (1\minb4) = 2 \neq 1 = \metr{\T}$.
           The miscomputation stems from counting \attr{b} twice.}
  \label{fig:bottom_up_DAG}
\end{figure}

\begin{example}
	\label{ex:bottom_up_DAG}
	\Cref{fig:bottom_up_DAG} shows how the bottom-up approach can fail when applied to \DAG-structured attack trees.
	Intuitively, the problem is that a visit to node $v$ in \Cref{alg:bottom_up_SAT}---or any bottom-up procedure that operates on the \AT structure---can only aggregate information on its descendants.
	So, the recursive call for $v$ cannot determine whether a sibling node in the \AT (i.e.\ any node $v'$ which is not an ancestor nor a descendant of $v$) shares a \BAS descendant with $v$.
	As a result, recursive computations for both $v$ and $v'$ may select a shared descendant $b\in\BAS$, and use \attr{b} in (both) their local computations.
	This causes the miscomputation in \Cref{fig:bottom_up_DAG}.
\end{example}


Workarounds to this issue include keeping track of the \BAS selected at each step by the metric \cite{KW18}, and operating on the \AT semantics \cite{MO06}.
In all cases the worst-case scenario has exponential complexity on the number of \AT nodes:
for \cite{KW18} this is in the input of the algorithm, i.e.\ determining the sets of necessary and optional clones; for \cite{MO06} and our \Cref{def:SAT:metric} the complexity lies in the computation of the semantics.

In general, one cannot hope for faster algorithms: \Cref{theo:NP_hard} shows that the core problem---computing minimal attacks of \DAG-structured attack trees---is NP-hard even in the simplest structure: plain attack trees with \AND/\OR gates.
The proof
\ifthenelse{\boolean{forCSF}}%
{
(provided in full in \cite{arXiv})
}{
(in \Cref{sec:proofs}, \cpageref{theo:NP_hard:proof})
}%
reduces the satisfiability of logic formulae in conjunctive normal form, to the computation of minimal attacks in general \SATs.

\begin{theorem}
	\label{theo:NP_hard}
	The problem of computing the smallest minimal attack
	of a DAG-structured static AT is NP-hard.
\end{theorem}

Note that the attribute domain $(\NN_\infty,\min,{+})$ allows for an attribution $\attrOp$, s.t.\ the \BAS that constitute the resulting metric can be extracted in polynomial time from its value.
This observation underpins the following \namecref{coro:NP_hard} of \Cref{theo:NP_hard}:

\begin{corollary}
	\label{coro:NP_hard}
	Computing  a metric for an attribute domain
	in a DAG-structured SAT is NP-hard.
\end{corollary}

\subsection{Binary decision diagrams}
\label{sec:SAT_DAGs:BDDs}

\BDDs offer an extremely compact representation of Boolean functions, whose size can grow linearly in the number of variables, i.e.\ the \BAS of the \AT \cite{Bry86}.
Although this depends on the variable ordering, and there exist functions where every \BDD is of exponential size, \DAG-structures that represent Boolean functions---such as fault trees and \ATs---often have small \BDD encodings \cite{BET13,RD97}.

A \BDD is a rooted \DAG \bddf that, intuitively, represents a Boolean function $f\from \BB^n\to\BB$ over variables $\Vars=\{x_i\}_{i=1}^n$.
The terminal nodes of \bddf represent the outcomes of $f$: $\bot$ or $\top$.
A nonterminal node $w\in\BDDnodes$ represents a subfunction $f_w$ of $f$ via its Shannon expansion.
That means that $w$ is equipped with a variable $\BDDlab(w)\in\Vars$ and two children:
$\low(w)\in\BDDnodes$, representing $f_w$ in case that the variable $\BDDlab(w)$ is set to $\bot$;
and $\high(w)$, representing $f_w$ if $\BDDlab(w)$ is set to $\top$.


\begin{definition}
	\label{def:BDD}
	A \emph{BDD} is a tuple $\bddT[]=(\BDDnodes,\low,\high,\BDDlab)$
	over a set \Vars where:
	\begin{itemize}[topsep=.5ex,parsep=.1ex,itemsep=0pt]
    \item	The \emph{set of nodes} \BDDnodes is partitioned into
			terminal nodes (\BDDnodesT) and
			nonterminal nodes (\BDDnodesN);
    \item	$\low    \from \BDDnodesN \to \BDDnodes$
			maps each node to its \emph{low child};
    \item	$\high   \from \BDDnodesN \to \BDDnodes$
			maps each node to its \emph{high child};
    \item	$\BDDlab \from \BDDnodes   \to \{\bot,\top\}\cup\Vars$
			maps terminal nodes to Booleans,
			and nonterminal nodes to variables:\\[.5ex]
			${\BDDlab(w)\in\begin{cases}
			\{\bot,\top\} & \text{if}~w\in\BDDnodesT,\\
			\Vars         & \text{if}~w\in\BDDnodesN.
			\end{cases}}$
	\end{itemize}
	Moreover, \bddT[] satisfies the following constraints:
	\begin{itemize}[topsep=.5ex,parsep=.1ex,itemsep=0pt]
	\item	$(\BDDnodes,\mathobject{E})$ is a connected \DAG, where\\
			\phantom{.}\hfill%
			$\mathobject{E}=\{(w,w')\in\BDDnodes^2 \mid w'\in\low(w)\cup\high(w)\}$;
	\item	\bddT[] has a unique root, denoted \BDDroot:\\
			\phantom{.}\hfill%
			$\exists!\,\BDDroot\in \BDDnodes.~%
			\forall w\in\BDDnodesN.~\BDDroot\not\in\low(w)\cup\high(w)$.
	\end{itemize}
\end{definition}

\paragraph{Reduced ordered BDDs}
We operate with \emph{reduced ordered BDDs}, simply denoted \BDDs.
This requires a total order ${<}$ over the variables.
For \Cref{def:BDD} this means that:
\begin{itemize}[label=\textbullet]
\item	\Vars comes equipped with a total order,
		so \bddf is actually defined over a pair \poset[\Vars][{<}];
\item	the variable of a node is of lower order than its children:
		$\forall\,w\in\BDDnodesN.\,%
			\BDDlab(w)<\BDDlab(\low(w)),\BDDlab(\high(w))$;
\item	the children of nonterminal nodes are distinct nodes;
\item	all terminal nodes are distinctly labelled.
\end{itemize}
This has the following consequences in the \BDD:
\begin{itemize}[label=\textbullet]
\item	there are exactly two terminal nodes:
		$\BDDnodesT=\{\oldbot,\!\oldtop\}$,
		with ${\BDDlab(\oldbot)=\bot}$ and ${\BDDlab(\oldtop)=\top}$;
\item	the label of the root node \BDDroot has the lowest order;
\item	in any two paths from \BDDroot to $\oldbot$ or $\oldtop$,
		the variables appear in the same (increasing) order.
\end{itemize}

\paragraph{Encoding static \ATs as \BDDs}
The key idea behind \BDDs is that evaluating a Boolean function $f$ on an input ${\bfx=(x_1,\ldots,x_n)\in\BB^n}$ is equivalent to following the corresponing path from \BDDroot to a terminal node:
when visiting node $w\in\BDDnodesN$ with $x_i=\BDDlab(w)$, the path goes to the child $\low(w)$ if $x_i=\bot$ in \bfx; else it goes to $\high(w)$.
The result $f(\bfx)\in\BB$ is the label of the terminal node reached.

This is used to encode fault trees as \BDDs via their structure function \cite{Rau93}, and extends to \ATs by letting ${\BAS=\Vars}$.
Technically, this exploits the Boolean function $\bfx\mapsto\sfunT(\attack_{\bfx})$, where the attack $\attack_{\bfx}$ contains the \BAS in whose position (determined by the total order ${<}$) the input $\bfx$ is $\top$.

Finally and importantly, since the metrics are defined on the set of minimal attacks of an AT \T, 
the \BDD \bddT must   exclusively represent the minimal attacks in \T.
This is achieved by using a variant of the Shannon expansion of the structure function \sfunT \cite{RD97}, which evaluates to $\top$ only when including the \BAS which are essential for the current attack under consideration.
Formally:
$
	\bfx \mapsto
		\big(x_1 \land f(\bfx_1) \land \neg f(\overbar{\bfx_1}) \big)
		\lor
		\big(\neg \overbar{x_1} \land f(\overbar{\bfx_1}) \big),
$
where one has $\bfx_1 \doteq (\top,x_2,\ldots,x_n)$ and $\overbar{\bfx_1} \doteq (\bot,x_2,\ldots,x_n)$.

\medskip%
\noindent%
\begin{minipage}{.72\linewidth}
\begin{example}
	\label{ex:SAT:BDD}
	Let $n<t<p$ in \sampleTs from \Cref{ex:running_examples}: the resulting \BDD $(\bddT[\sampleTs])$ is illustrated to the right.
	As usual, the children of a node appear below it (so the root node is on top), and a dashed line from $w$ to a child $w'$ means that $w'=\low(w)$, and a solid line means that $w'=\high(w)$.
\end{example}
\end{minipage}
\quad
\begin{minipage}{.2\linewidth}
\includegraphics[width=.95\linewidth]{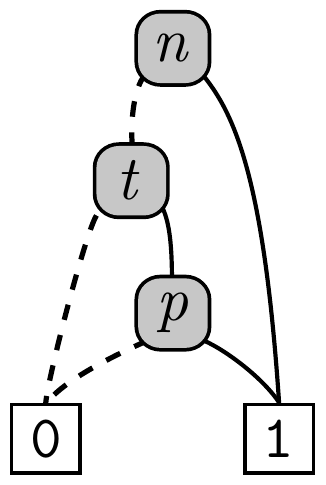}
\end{minipage}
\medskip


\subsection{BDD-based algorithm for DAG-structured SATs}
\label{sec:SAT_DAGs:algorithm}

\Cref{alg:bottom_up_BDD} computes metrics for \DAG-structured trees given an attribute domain $\ndomain=(\Vdom,\operOR,\operAND,\ntOR,\ntAND)$, where ${\ntOR,\ntAND\in\Vdom}$ are neutral elements for $\operOR$ and $\operAND$:
$\forall x\in\Vdom.~\ntOR \operOR x = \ntAND \operAND x = x$.
Also and just like \BUSAT, \Cref{alg:bottom_up_BDD} requires $(\Vdom,\operOR,\operAND)$ in \ndomain to be a semiring attribute domain.

It is common for semiring definitions to require the presence of neutral elements \cite{Mac71}: they are needed for \DAG-structured \SATs, although not for tree-structured \SATs.
Examples of neutral elements in \Cref{tab:SAT:metric} are $\ntOR=\infty$ and $\ntAND=0$ for min cost, and $\ntOR=0$ and $\ntAND=1$ for (max) discrete probability.

\paragraph{The algorithm}
The idea behind \Cref{alg:bottom_up_BDD} is to traverse the \BDD top-down (or, equivalently, bottom-up), accumulating via $\operAND$ the values of the \BASs included in the attack under consideration.
For that, at each node $w$ visited in the \BDD \bddT, \BUBDD recursively computes the metric value \metr{\T[w]} for the \AT whose \BDD \bddT[{\T[w]}] is the sub-\BDD of \bddT with root $w$.
So, starting at the root \BDDroot of the \BDD \bddT, algorithm \BUBDD considers the only two possible types of attack:
\begin{itemize}[leftmargin=1em]
\item	Those that include $\BDDlab(\BDDroot)=v\in\BAS$:
		\begin{itemize}
		\item	the metric for this suite of attacks is computed in a
				recursive call of \BUBDD on the child $\high(\BDDroot)=h$;
		\item	these attacks use $v\in\BAS$ so their metrics use
				${\attr{v}\in\Vdom}$, accumulated via $\operAND$
				(which distributes over $\operOR$);
		\item	the result is $\metr{\T[h]}\operAND\attr{v}\in\Vdom$,
				where \bddT[{\T[h]}] represents the suite of attacks
				of \T that require $v$ to succeed.
		\end{itemize}
\item	Those that exclude $\BDDlab(\BDDroot)$:
		\begin{itemize}
		\item	the metric is computed by recursion on $\low(\BDDroot)=\ell$;
		\item	these attacks exclude $v$ and therefore do not use \attr{v};
		\item	the result is $\metr{\T[\ell]}\in\Vdom$,
				where \bddT[{\T[\ell]}] represents the suite of successful
				attacks	of \T that exclude $v$.
		\end{itemize}
\item	The final metric for \T is the disjunction
		of these the two recursive calls:
		$\metr{\T[\ell]} \operOR \big( \metr{\T[h]}\operAND\attr{v} \big)$.
\item	The base cases of the recursive calls are the \BDD leaves:
		\begin{itemize}
		\item	$\BDDlab(\oldbot)=\bot$ is given the neutral element
				$\ntOR\in\Vdom$;
		\item	$\BDDlab(\oldtop)=\top$ is given the neutral element
				$\ntAND\in\Vdom$.
		\end{itemize}
%
\end{itemize}
The pseudocode of this procedure is given as \Cref{alg:bottom_up_BDD}.


\begin{algorithm}
	\KwIn{\BDD $\bddT=(\BDDnodes,\low,\high,\BDDlab)$,\newline
	      node $w\in\BDDnodes$,\newline
	      attribution $\attrOp$,\newline
	      \mbox{semiring attribute domain 
	      $\ndomain=(\Vdom,\operOR,\operAND,\ntOR,\ntAND)$.}}
	\KwOut{Metric value $\metr{\T}\in\Vdom$.}
	\BlankLine
	\uIf{$\BDDlab(w)=\bot$}{%
		\Return{\ntOR}
	} \uElseIf{$\BDDlab(w)=\top$}{%
		\Return{\ntAND}
	} \Else(\tcp*[h]{$w\in\BDDnodesN$}) {%
		\Return{${\BUBDD(\bddT,\low(w),\attrOp,\ndomain)} \operOR
			{\big(\BUBDD(\bddT,\high(w),\attrOp,\ndomain) \operAND
				\attr{\BDDlab(w)}\big)}$}
	}
	\caption{\BUBDD for a \DAG-structured \SAT \T}
	\label{alg:bottom_up_BDD}
\end{algorithm}

\medskip
  
\noindent%
\begin{minipage}{.72\linewidth}
\begin{example}
	\label{ex:SAT:bottom_up_BDD}
	For the \DAG-structured \SAT shown in \Cref{fig:bottom_up_DAG}, the order ${b<a<c}$ of its \BAS yields the \BDD to the right.
	To compute the min cost (like in \Cref{fig:bottom_up_DAG}) we employ the attribution $\attrOp=\{{a\mapsto3},{b\mapsto1},{c\mapsto4}\}$ and the domain $(\NN_\infty,\min,+)$.
	Moreover, to use \Cref{alg:bottom_up_BDD}, we choose the neutral elements 
\end{example}
\end{minipage}
\quad
\begin{minipage}{.2\linewidth}
\includegraphics[width=.95\linewidth]{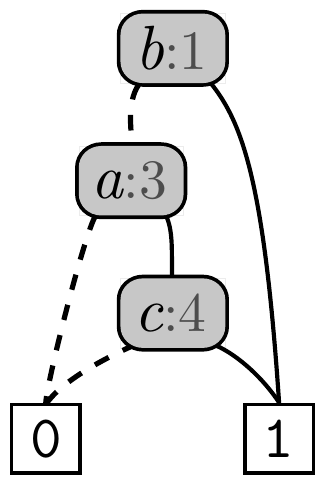}
\end{minipage}
\par\vspace{.6ex plus .25ex minus .15ex}\noindent%
$\ntOR=\infty$ for $\min$ and $\ntAND=0$ for ${+}$, i.e.\ we use the attribute domain $\ndomain=(\NN_\infty,\min,{+},\infty,0)$.
Let the nonterminal nodes of the \BDD \bddT be $\BDDnodesN=\{w_a,w_b,w_c\}$.
For $w\in\BDDnodes$ let $\BU(w)=\BUBDD(\bddT,w,\attrOp,\ndomain)$, then we compute the metric:

\vspace{-2ex}
\begingroup
\allowdisplaybreaks
\begin{align*}
  \BU(\BDDroot)
	&= \BU(w_a) \minb \big(\BU(\oldtop)+\attr{b})\\
	&= \BU(w_a) \minb \big(\ntAND+1)\\
	&= \BU(w_a) \minb 1\\
	&= \big(\BU(\oldbot) \minb (\BU(w_c)+\attr{a}) \big) \minb 1\\
	&= \big(\ntOR \minb (\BU(w_c)+3) \big) \minb 1\\
	&= \big(\BU(w_c)+3 \big) \minb 1\\
	&= \big((\BU(\oldbot)\minb(\BU(\oldtop)+\attr{c}))+3 \big) \minb 1\\
	&= \big((\ntOR\minb(\ntAND+4))+3 \big) \minb 1\\
	&= (4+3) \minb 1 ~=~ 1.
\end{align*}
\endgroup
\vspace{-2ex}

To compute instead the (max) discrete probability we use the attribution $\attrOp'=\{{a\mapsto0.1},{b\mapsto0.05},{c\mapsto0.6}\}$ and the attribute domain $\ndomain'=({[0,1]_\QQ},\max,\ast,0,1)$.
Then computations are as before until the last line, which here becomes: ${(\attrOp'(c)\ast\attrOp'(a))\maxb\attrOp'(b)} = {(0.6\ast0.1)\maxb0.05} = 0.06$.

\Cref{theo:bottom_up_BDD} states the correctness of \Cref{alg:bottom_up_BDD}, i.e.\ that it yields the metric for a static \AT given in \Cref{def:SAT:metric} regardless of its structure.
We prove \Cref{theo:bottom_up_BDD} in
\ifthenelse{\boolean{forCSF}}%
{
\cite{arXiv}, by induction in the number of levels of the \BDD \bddT: this is the cardinality of its set of nodes \BDDnodes, whose labels are the \BASs of \T.
Our proof relies on the fact that the leaf $\oldtop$ in \bddT cannot be the $\low$ child of a node, and analogously $\oldbot$ cannot be a $\high$ child.
The intuition behind this is that visiting $\low(w)$ symbolises the act of excluding the node $\BDDlab(w)\in\BAS$ from an attack.
Since static \ATs are coherent, \emph{excluding a BAS cannot make an attack succeed.}
Therefore, taking the $\low$ child of $w$ cannot lead to $\oldtop$; the reasoning for $\oldbot$ and $\high$ is analogous.
}{
\Cref{sec:proofs}, \cpageref{theo:bottom_up_BDD:proof}.
}%

\begin{theorem}
	\label{theo:bottom_up_BDD}
	\def\root{\ensuremath{\BDDroot[{\bddT}]}\xspace}
	\def\IC#1{\ensuremath{\mathrm{IC}_{#1}}\xspace}
	Let \T be a static \AT, \bddT its \BDD encoding over \poset[\BAS][{<}],
	$\attrOp$ an attribution on \Vdom,
	and $\ndomain=(\Vdom,\operOR,\operAND,\ntOR,\ntAND)$ an attribute domain
	with neutral elements resp.\ for $\operOR$ and $\operAND$.
	Then $\metr{\T} = \BUBDD(\bddT,\root,\attrOp,\ndomain)$.
\end{theorem}

\paragraph{\BDDs to compute semantics}
The \BDD encoding of a static \AT \T can also be used to compute \ssem{\T}.
Consider a path $\pi=a_1\overbar{a_2}\cdots a_\ell$ from the root of \bddT to its $\oldtop$-leaf:
$a_i$ (resp.\ $\overbar{a_i}$) denotes that $\pi$ goes to the $\high$ (resp.\ $\low$) child of the \BDD node labelled with $a_i\in\BAS$.
Then $\pi$ represents a successful attack $\attack \doteq \{ a_i\in\BAS \mid a_i~\text{appears positive in}~\pi \} \in \allAttacks_{\T}$.
To compute all successful attacks:
\begin{enumerate*}
\item	find all distinct paths $\{\pi_j\}_{j=1}^n$ in the graph of
		\bddT, from its root node to its $\oldtop$-leaf;
\item	let $A_j=\{\text{positive}~\BAS~\text{in}~\pi_j\}$.
\end{enumerate*}
Providing that \bddT encodes minimal attacks only, the result is $\{A_j\}_{j=1}^n=\ssem{\T}$.

\subsection{Computing the $k$-top metric values}

The approach described above can be extended to efficiently compute the $k$-top values for a given metric.
This problem asks not only the min/max value of the metrics from \Cref{tab:SAT:metric}, but also the next $k-1$ min/max values, e.g.\ the cost of the $k$ cheapest attacks, or the probability of the $k$ most likely ones.

Such $k$-top values can be computed by weighing the $\high$ edges of the \BDD with their corresponding (source-) \BAS attributes, and finding the $k$-shortest-weighted paths from the root of the \BDD to its $\oldtop$-leaf.
We present this idea as \Cref{alg:shortest_path_BDD}.

\Cref{alg:shortest_path_BDD} relies on an implementation of $\kshortest$: the $k$-shortest-paths algorithm for \DAGs.
This is a well-known extension of the Dijkstra (or Thorup) algorithm \cite{Dij59,Tho99}.
For a \DAG $G$ with edges weighted by the matrix $Q$, $\kshortest(G,Q,s,t,k,\circ)$ returns the weight of the $k$-shortest paths from a (source) node $s$ of $G$, to a (target) node $t$, using operator $\circ$ to accumulate weight.


\begin{algorithm}
	\def\matrix{\mathobject{Q}}
	\def\sgn{\ensuremath{\mathit{sgn}}}
	\KwIn{\BDD $\bddT=(\BDDnodes,\low,\high,\BDDlab)$,\newline
	      number of values to compute $k\in\NN$,\newline
	      attribute domain $\domain=(\Vdom,\operOR,\operAND)$,\newline
	      attribution $\attrOp$.}
	\KwOut{$k$-top metric values of \T for $\attrOp$ and $\domain$.}
	\BlankLine
	\matrix\ := $0$-filled $\card{\BDDnodes}\times\card{\BDDnodes}$ matrix
		\\ 
	\leIf(\tcp*[f]{$\operOR=\max$})
		{$\operOR=\min$}  
		{\sgn\ := $1$}    
		{\sgn\ := $-1$}   
	\ForEach{nonterminal node $w\in\BDDnodesN$}{%
		$\matrix[w][\high(w)]$ := $\sgn\ast\attr{\BDDlab(w)}$\\
	}
	\Return{$\sgn\ast\kshortest(\bddT,\matrix,\BDDroot[\bddT],\oldtop,k,\operAND)$}
	\caption{$\mathtt{k\_top}$ metric values for a \SAT \T}
	\label{alg:shortest_path_BDD}
\end{algorithm}

\Cref{alg:shortest_path_BDD} works for $\operOR\in\{\min,\max\}$, using a sign change to compute max-top values, in which case the implementation of $\kshortest$ must support negative weights.
The correctness of the \namecref{alg:shortest_path_BDD} is a direct consequence of the (correct) encoding of the minimal attacks of \T by the \BDD \bddT, and the $\kshortest$ algorithm.

\begin{example}
	\label{ex:k-top_values}
	Consider the \DAG-structured \SAT from \Cref{fig:bottom_up_DAG}, $\T={\AND\big(\OR(a,b),\OR(b,c)\big)}$.
	To compute its 2 cheapest attacks under the attribution $\attrOp = \{{a\mapsto3}, {b\mapsto1}, {c\mapsto4}\}$, let $b<a<c$ s.t.\ \bddT is as in \Cref{ex:SAT:bottom_up_BDD}.
	The $\low$ edge of the root $b$ (that encodes ``not performing $b$'') is labelled with cost $0$, and the $\high$ edge with cost $\attr{b}=1$; the same is done for $a$ and $c$.
	Then the shortest-weight path from the root of \bddT to its $\top$-labelled leaf is $\pi_1=b$, which yields the cheapest attack $\attack_1=\{b\}$ with cost $\metrA{\attack_1}=\attr{b}=1$.
	Second to that we find the path $\pi_2=\overline{b}ac$, which yields the second-cheapest attack $\attack_2=\{a,c\}$ with cost $\metrA{\attack_2}=\attr{a}\operAND\attr{c}=3+4=7$.
\end{example}

\section{Analysis of Dynamic Attack Trees}
\label{sec:DAT}

In the presence of \SAND gates, the execution order of the \BAS becomes relevant.
This affects primarily the semantics, i.e.\ what it means to perform a successful attack, but also security metrics become sensitive to the sequentiality of events.

\subsection{Partially-ordered attacks and well-formedness}
\label{sec:DAT:wellformed}

As for the static case, the semantics of a dynamic attack tree (\DAT) is defined by its successful attack scenarios.
However, \DATs necessitate a formal notion of order, because a sequential gate $\SAND(v_1,\ldots,v_n)$ succeeds only if every $v_i$ child is completely executed before $v_{i+1}$ starts.


Such constructs model dependencies in the order of events.
E.g.\ in H\aa{}stad's broadcast attack, $n$ messages must first be intercepted, from which an $n$-th root (the secret key) may be computed.
In this standard \emph{ordered interpretation}, an activated \BAS is uninterruptedly completed.
This rules out constructs that introduce circular dependencies such as $\SAND(a,b,a)$.\!%
\footnote{\,Cf.\ Kumar et al. (2015), who separates activation from execution of a \BAS and can therefore operate with $\SAND(a,b,a)$ \cite{KRS15}.}

Therefore, an attack scenario that operates with \SAND gates is not just a set $\attack\subseteq\BAS$, but rather a partially-ordered set: a \emph{poset} \poset, where $a\prec b$ indicates that $a\in\attack$ must be carried out strictly before $b\in\attack$. 
Incomparable basic attack steps can be executed in any order, or in parallel.

Thus, the attack \poset indicates that all \BAS in \attack must be executed, and their execution order will respect ${\prec}$\,.
This succinct construct can represent combinatorially many execution orders of \BAS.
For instance \poset[\{a,b\}][\{(a,a),(b,b)\}] allows three executions: the sequence $(a,b)$, and $(b,a)$, and the parallel execution $a\|b$.
Instead, \poset[\{a,b\}][\{(a,a),(b,b),(a,b)\}] only allows the execution sequence $(a,b)$.

Partial orders are reflexive and transitive, so for instance $\SAND(a,b,c)$ gives rise to ${\prec}=\{(a,a),(b,b),(c,c),$ $(a,b),(b,c),(a,c)\}$.
We use an abbreviated notation that depicts their transitive reduction,
so the previous case becomes $\{{a\prec b},{b\prec c}\}$.
This is a textual equivalent to the (unique) Hasse diagram that represents the poset.

\begin{example}
	\label{ex:DAT:poset}
	Consider the dynamic attack tree from \Cref{fig:AT:example:dynamic}: $\sampleTd=\OR\big(\AND(\ff,\ww),\SAND(\ww,\cc)\big)$.
	The posets \poset[\{\ww,\cc\}][\{\ww\prec\cc\}] and \poset[\{\ff,\ww\}][\emptyset] are attack scenarios for \sampleTd, where ${\prec}=\emptyset$ in the latter implies that \ff and \ww can be executed in any order, even in parallel.
	So this poset represents (among others) the \BAS execution sequence $(\ww,\ff)$, which results in a \TLA of \sampleTd.
	Similarly, the poset \poset[\{\ww,\cc\}][\emptyset] allows the execution sequence $(\cc,\ww)$: this violates the gate $\SAND(\ww,\cc)$ so it cannot be considered a valid attack for \sampleTd.
\end{example}

Since successful attacks \poset must ensure all  the sequential orders imposed by \SAND gates, it is possible to 
express infeasible requirements. For example, $\SAND(a,b,a)$ indicates that $a$ must precede $b$, and $b$ must precede $a$. 
To rule out these cases, we operate with well-formed \DATs only. 
A \DAT is well-formed if, for every $\SAND(v,v')$, all the \BASs below $v$ are executed before any of the \BASs below $v'$. 

\begin{definition}[Well-formedness]
	\label{def:wellformed}
	The \emph{BAS descendants} of a node $v\in\ATnodes$
	are 
	$\desc(v)=\{v\}$ if $\typOp(v)=\tBAS$, and
	$\desc(v)=\bigcup_{u\in\chOp(v)}\desc(u)$ otherwise.
	The \emph{ordering graph} of \T is the directed graph
	${\ograph{\T}=\big(\BAS_{\T},\before[\T]\big)}$ s.t.\ $a\before[\T]b$
	iff there is a \SAND gate $v=\SAND(v_1,\ldots,v_n)$ with 
	$a\in\desc(v_i)$ and $b\in\desc(v_{i+1})$ for some $0<i<n$. 
	\T is \emph{well-formed} if \ograph{\T} is acyclic; otherwise
	\T is \emph{ill-formed}.
\end{definition}

\begin{example}
	\label{ex:wellformed}
	\Cref{fig:wellformed:invalid} presents two ill-formed dynamic \ATs: \T[1] and \T[2].
	In contrast, \T[3] (\Cref{fig:wellformed:valid}) and \sampleTd (\Cref{fig:AT:example:dynamic}) are examples of well-formed dynamic attack trees.
\end{example}
	
	Our well-formedness criterion can rule out \DATs for which successful sequential attacks do exist.
	In \Cref{fig:wellformed:invalid}, the execution sequence $(a,b)$ makes the \TLA of \T[2] succeed.
	But \T[2] is a modelling error under our ordered interpretation of \SAND gates, because its subtree $\SAND(b,a)$ indicates that $b$ must be completed to enable $a$.
	Nevertheless, such execution makes sense under an interpretation of the \OR gate that allows the parallel execution of both children, and sees who finishes first. 
	To cover these cases, future work can relax our assumptions.

\begin{figure}
  \vspace{-1ex}
  \centering
  \begin{subfigure}[t]{.4\linewidth}
	\includegraphics[height=2cm]{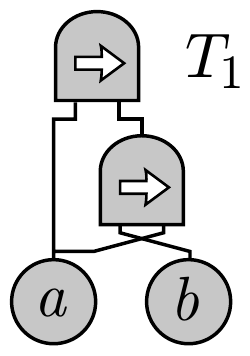}
	~
	\includegraphics[height=2cm]{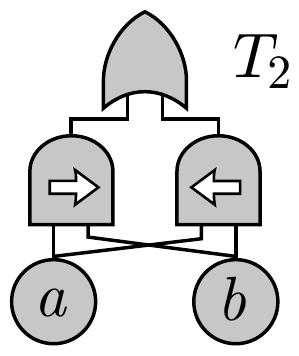}
	\caption{Ill-formed \ATs}
	\label{fig:wellformed:invalid}
  \end{subfigure}
  \qquad
  \begin{subfigure}[t]{.3\linewidth}
	\centering
	\includegraphics[height=2cm]{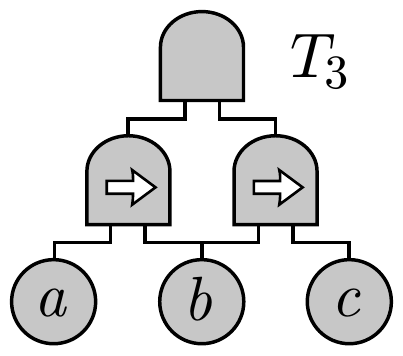}
	\caption{Well-formed \AT}
	\label{fig:wellformed:valid}
  \end{subfigure}
  ~\\[1.5ex]
  \begin{subfigure}[t]{\linewidth}
	\centering
	\includegraphics[width=.35\linewidth]{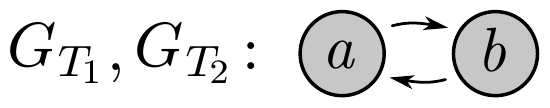}
	\quad~
	\includegraphics[width=.35\linewidth]{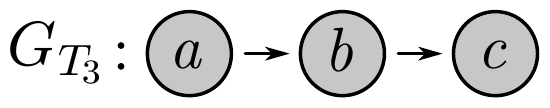}
	\caption{Ordering graphs, with transitive reduction of $\before$}
	\label{fig:wellformed:ographs}
  \end{subfigure}
  \caption{Well-formedness of dynamic Attack Trees}
  \label{fig:wellformed}
  \vspace{-3ex}
\end{figure}

\subsection{Semantics for dynamic attack trees}
\label{sec:DAT:semantics}

The transitive reduction of
the ordering graph \ograph{\T} is a Hasse diagram, that represents the poset of all \BAS nodes and \SAND gates of \T.
This matches the notion of poset that has been intuitively introduced as an attack, and that we formalise in \Cref{def:DAT:attack}.
This \namecref{def:DAT:attack} also lifts the \emph{successful} and \emph{minimal attacks} of \SATs to the category of posets.
The resulting notion of dynamic attack, which underpins our \DAT semantics, can thus be seen as an extension of the standard concepts for \SATs, conservative w.r.t.\ our notion of well-formedness (\Cref{def:wellformed}).

Interestingly, well-formedness plus the structure function of static \ATs suffices to define successful attacks in \DATs: these must
(1) respect all \SAND gates, and
(2) be successful in the corresponding static \AT, obtained by transforming \SAND gates into \AND{s}.
As a consequence, we need not introduce a new structure function for dynamic \ATs:

\begin{definition}[Attacks in dynamic \ATs]
	\label{def:DAT:attack}
	Let \T be a well-formed \DAT with ordering graph $\ograph{}=(\BAS,\before)$:
	\begin{itemize}[topsep=.5ex,parsep=.1ex,itemsep=0pt]
	\item	An \emph{attack scenario}, or shortly an \emph{attack},
			of \T is a poset \poset s.t.\ $\attack\subseteq\BAS$,
			and ${\prec}={\before}\restriction_{\attack}$
			restricts the edge relation $\before$ to \attack, i.e.\
			$\forall a,b\in \attack.\,(a\prec b)\Leftrightarrow(a\before b)$;
	\item	An attack \poset is \emph{successful} if
			$\sfunT[\T'](\attack)=\top$, where $\sfunT[\T']$ is the
			structure function of the \SAT $\T'$, which is obtained by
			replacing every \SAND gate in \T by an \AND;
	\item	A successful attack \poset is \emph{minimal}
			if both $\attack$ and $\prec$ are minimal,
			i.e.\ $\nexists\,\text{successful}\:\poset[\attack'][\prec'].\,
			(\attack'\subsetneq\attack) \lor ({\prec}\subsetneq{\prec'})$.
	\end{itemize}
\end{definition}

\begin{example}
	\label{ex:DAT:attack}
	The ordering graph of the dynamic attack tree \sampleTd from \Cref{fig:AT:example:dynamic} 
has the single edge $\ww\before\cc$.
	Therefore, three successful attacks for \sampleTd are: \poset[\{\ww,\cc\}][\{\ww\prec\cc\}], \poset[\{\ff,\ww\}][\emptyset], and \poset[\{\ff,\ww,\cc\}][\{\ww\prec\cc\}].
	The first two are minimal attacks.
	Instead, the attack \poset[\{\ff,\cc\}][\emptyset] is not successful, and the poset \poset[\{\ww,\cc\}][\{\cc\prec\ww\}] is not an attack since $(\cc,\ww) \in {\prec}\setminus{\before}$, where $\setminus$ denotes set difference.
\end{example}

In minimising also over the partial order ${\prec}$\,, \Cref{def:DAT:attack} makes minimal attacks the least restrictive in terms of sequential dependencies.
Moreover, an \emph{attack suite} \suite of a dynamic \AT \T is a set of attacks, just like for \SATs.
Also $\allAttacks_{\T}$ denotes the universe of attacks of \T, and $\allSuites_{\T}$ its universe of attack suites.

Unlike for \SATs, however, the execution order imposed by \SAND gates makes dynamic \ATs \emph{non-coherent} in general.
Consider ${\SAND(a,\OR(b,c))}$, where \poset[\{a,b\}][\{a\prec b\}] is a successful attack but \poset[\{a,b,c\}][\{c\prec a,a\prec b\}] is not, even though the latter extends the former with ${c\in\BAS}$.

Coherence is a desired property: it means that all successful attacks of a tree are characterised by its minimal attacks.
To maintain this property in the presence of \SAND gates, \Cref{def:DAT:attack} forces the partial order of an attack \poset to be a restriction (to \attack) of the edge relation of the corresponding ordering graph.
Posets that either \emph{omit a required execution order} (e.g.\ the last one in \Cref{ex:DAT:poset}), or \emph{add an invalid execution order} (e.g.\ the last one in \Cref{ex:DAT:attack}), are not attacks of \T. 
This restriction in \Cref{def:DAT:attack} results in the coherence of \DATs:

\begingroup
\def\A{\attack_1}
\def\P{{\prec_1}}
\def\AA{\attack_2}
\def\PP{{\prec_2}}
\begin{proposition}
	\label{prop:coherent_DATs}
	A well-formed dynamic \AT \T is coherent: 
	if $\poset[\A][\P],\poset[\AA][\PP] \in \allAttacks_{\T}$
	and \poset[\A][\P] is a successful attack, then $\A\subseteq\AA$
	implies that \poset[\AA][\PP] is also a successful attack.
\end{proposition}
\begin{proof}
	Let $\poset[\A][\P],\poset[\AA][\PP] \in \allAttacks_{\T}$.
	By \Cref{def:DAT:attack}, if \poset[\A][\P] is a successful attack of \T 
	then $\sfunT[\T'](\A)=\top$, where $\T'$ is the static \AT obtained by
	transforming all \SAND gates of \T to \AND{s}.
	Since \SATs are coherent:
	${\A\subseteq\AA\Rightarrow\sfunT[\T'](\AA)=\top}$.
	Finally by \Cref{def:wellformed,def:DAT:attack},
	$\PP={\before[\T]}\restriction_{\AA}$ implies that the sequences
	of execution of $\AA\subseteq\BAS$ represented by \poset[\AA][\PP]
	respect the order imposed by the \SAND gates of \T.
\end{proof}
\endgroup

This means that, analogously to static \ATs, the semantics of dynamic \ATs can be given by their minimal attacks:

\begin{definition}
	\label{def:DAT:semantics}
	The \emph{semantics of a well-formed \DAT} \T,
	denoted $\dsem{\T}\in\allSuites_{\T}$, is its suite of minimal attacks.
	%
\end{definition}

A price to pay for this result, and for such straightforward extensions of static concepts, is a strict notion of well-formed dynamic \AT:
besides discarding modeling errors such as \T[2] in \Cref{fig:wellformed:invalid}, it also discards \DATs where the children of a \SAND gate share subtrees.
To see this let ${\T = \SAND(v_1,v_2)} = {\SAND\big(\AND(a,b),\AND(b,c)\big)}$, whose ordering graph \ograph{\T} has edges from every descendant $\desc(v_1)=\{a,b\}$ to every descendant ${\desc(v_2)=\{b,c\}}$.
But then \ograph{\T} has a self-loop in the \BAS $b\before b$, which means that \T is ill-formed.

Our semantics also entails a strict notion of (successful) attack, that rules out some interleavings in the execution of high-level \SAND gates.
Consider e.g.\ ${\T'=\SAND\big(a,\AND(b,c)\big)}$, where \Cref{def:wellformed} forces $a$ to occur before any of $\{b,c\}$.
Then $b\before a$ is \emph{not an edge} in \ograph{\T'}, so our attacks \emph{exclude the order $b\prec a$}, even though $(b,a,c)$ is a valid execution sequence in $\T'$.

To relax this we need a more complex notion of ordering graph, as we discuss in \Cref{sec:conclu}.
However, this work is about the efficient computation of metrics, and as we show next these metrics are invariant for the different valid orders of execution of \BAS.
Therefore, here we use the stricter but simpler semantics that stems from \Cref{def:wellformed} of well-formedness.

Finally, \Cref{lemma:dsem} characterises the semantics resulting from \Crefrange{def:wellformed}{def:DAT:semantics}, analogously to how \Cref{lemma:ssem} does it for static \ATs.
This is key to prove the correctness of linear-time algorithms that compute metrics on tree-structured \DATs.
\ifthenelse{\boolean{forCSF}}{}{We prove this \namecref{lemma:dsem} in \Cref{sec:proofs}, \cpageref{lemma:dsem:proof}.}

\begin{lemma}
	\label{lemma:dsem}
	\def\REF#1{\textit{\ref{#1})}}
	\def\REFS#1#2{\textit{\ref{#1})--\ref{#2})}}
	Consider a well-formed \DAT with nodes ${a\in\BAS}$, ${v_1,v_2\in\ATnodes}$,
	that has a proper tree structure. Then:
	\begin{enumerate}
	\item	$\dsem{a} = \{ \poset[\{a\}][\emptyset] \}$;%
			\label{lemma:dsem:BAS}
	\item	$\dsem{\OR(v_1,v_2)} = \dsem{v_1} \cup \dsem{v_2}$;%
			\label{lemma:dsem:OR}
	\item	\rmkCEB{this $i$-subs\-cript notation (e.g.\ \dsem{v_i})
					isn't 100\% mathematically correct. We may restate
					\Cref{lemma:dsem} in a journal version with space}
			$\dsem{\AND(v_1,v_2)} = \mbox{$\left\{
				\poset[\attack_1{\cup}\attack_2][{\prec_1}{\cup}{\prec_2}]
				\,|\, \poset[\attack_i][\prec_i]\in\dsem{v_i}
			\right\}$}$;%
			\label{lemma:dsem:AND}
	\item	$\dsem{\SAND(v_1,v_2)} = \{
				\poset[\attack_1\cup\attack_2 \,]%
				      [~{\prec_1}\cup{\prec_2}\cup{\attack_1\times\attack_2}]
				\cdots$ \\
				\hspace*{\stretch{1}} $\cdots
					\mid \mbox{$\poset[\attack_i][\prec_i]\in\dsem{v_i}$}
			\}$;%
			\label{lemma:dsem:SAND}
	\item	In cases \REFS{lemma:dsem:OR}{lemma:dsem:SAND} above the
			\dsem{v_i} are disjoint, and in cases \REF{lemma:dsem:AND} and
			\REF{lemma:dsem:SAND} moreover the $A_i$ are pairwise disjoint.
			\label{lemma:dsem:disjoint}
	\end{enumerate}
\end{lemma}

\ifthenelse{\boolean{forCSF}}{%
\def\REF#1{\textit{\ref{#1})}}
We prove this \namecref{lemma:dsem} in \cite{arXiv}: just like for \Cref{lemma:ssem}, we use induction in the structure of an \AT whose root is the node $v$ on the left-hand side \dsem{v} of the equalities.
In particular, cases \REF{lemma:dsem:BAS}, \REF{lemma:dsem:OR}, and \REF{lemma:dsem:disjoint}, are trivial extensions of \Cref{lemma:ssem};
case \REF{lemma:dsem:AND} uses the minimality of the $\prec$ relation;
and case \REF{lemma:dsem:SAND} moreover considers the order requirements imposed by the \SAND gate on the \BAS descendants of $v_1$ and $v_2$.
}{}

\paragraph{Comparison with literature}
The semantics for dynamic \ATs resulting from \Crefrange{def:wellformed}{def:DAT:semantics} resembles the so-called \emph{series-parallel graphs} from \cite{JKM+15}.
We define dynamic attacks as posets for a number of reasons:
\begin{itemize}[label=\textbullet,leftmargin=1em]
\item	they are a succinct, natural lifting of the \SAT concepts,
		that facilitate the extension of earlier results such as the
		characterisation of \dsem{\cdot} in \Cref{lemma:dsem};
\item	metrics can be formally defined on this semantics,
		decoupling specific algorithms from a notion of correctness;
\item	in particular, this allows us to define algorithms to compute metrics
		regardless of the tree- or \DAG-structure of the \DAT.
\end{itemize}
%
%
\begingroup
\def\series{\boldsymbol{\cdot}}
The latter is different for \cite{JKM+15}, which does not work for \DAG-structured \DATs as noted in \cite{KW18}.
This can be illustrated in $\T[3] = \AND\big(\SAND(a,b),\SAND(b,c)\big)$, the \AT from \Cref{fig:wellformed:valid} whose series-parallel graph is $\mathit{SP}_3 = {(a\series b)\parallel(b\series c)}$.
Attributes and metrics are also defined in \cite{JKM+15}, choosing operators for \AND and \SAND gates which are resp.\ mapped to $\parallel$ and $\series$ in $\mathit{SP}$.
Let the operator be ${+}$, e.g.\ to compute attack cost, and consider the attribution $\attrOp=\{{a\mapsto1},{b\mapsto4},{c\mapsto8}\}$:
the metric obtained for $\mathit{SP}_3$ is $(1+4)+(4+8) = 17$.
But the expected result is $13$, i.e.\ execute every \BAS once.
\endgroup
\\[.6ex]
In contrast, posets entail a formal definition of metric over \DAT semantics---given now in \Cref{sec:DAT:metrics}---which in particular yields the expected result even for \DAG-structured \DATs.

\subsection{Security metrics for dynamic attack trees}
\label{sec:DAT:metrics}

The same fundamental concepts of metric for static \ATs work for dynamic \ATs: from the attributes of every \BAS, obtain a metric for each attack in \dsem{\T}, and from these values compute the metric for \T.
Thus, the generic notion of metric given by \Cref{def:metric} in \Cref{sec:SAT:metrics} carries on to this \namecref{sec:DAT:metrics}.

However, attribute domains do not suffice for \DATs: metrics such as min attack time are sensitive to order dependencies among \BAS.
This requires an additional \emph{sequential operator ${\operSAND \from \Vdom^2 \to \Vdom}$}, to compute values of sequential parts in an attack.
Therefore, metric computations gain an extra step:
\begin{enumerate}
\setcounter{enumi}{-1}
\item	first, an attribution $\attrOp$ assigns a value to each \BAS;
\item	then, a sequential metric $\metrSOp$ uses the operator $\operSAND$ to assign a value to each sequential part of an attack;
\item	then, a parallel metric $\metrAOp$ uses $\operAND$ to assign a value to each attack, as the parallel execution of all its sequential parts;
\item	finally, the metric $\metrOp$ uses $\operOR$ to assign a value to the whole attack suite, by considering all its constituting attacks.
\end{enumerate}

This can be pictured on the Hasse diagrams that represent the posets:
for every attack $\poset\in\dsem{\T}$, its (unique) Hasse diagram \Hasse is the restriction of the ordering graph \ograph{\T} to the nodes in $\attack\subseteq\BAS$---see e.g.\ \Cref{fig:Hasse} for \sampleTd from \Cref{ex:running_examples}.
So \Hasse is a set of nodes, some of which are connected by edges and form a connected component \ccomp.
In the 4-steps computation described above, this means that:
\begin{enumerate}
\setcounter{enumi}{0}  
\item	$\metrSOp$ uses $\operSAND$ on each connected component
		$\{\ccomp_i\}_{i=1}^{n_{\attack}}$ of \Hasse,
		yielding one value $s_i\in\Vdom$ for each $\ccomp_i$;
\item	$\metrAOp$ uses $\operAND$ on $\{s_i\}_{i=1}^{n_{\attack}}$,
		yielding a metric for the attack \Hasse;
\item	$\metrOp$ uses $\operOR$ on the metrics of all attacks in \dsem{\T},
		yielding the metric for the dynamic attack tree \T.
\end{enumerate}

\begin{figure}
  \centering
  \def\WIDTH{.16\linewidth}
  \def\OG{\ograph{\scalebox{.8}{$\sampleTd$}}}
  \begin{subfigure}{\WIDTH}
	\centering
	\includegraphics[width=\linewidth]{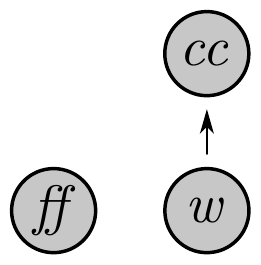}
	\caption{\OG}
	\label{fig:Hasse:ograph}
  \end{subfigure}
  \quad~
  \begin{subfigure}{\WIDTH}
	\centering
	\includegraphics[height=1.02\linewidth]{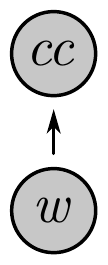}
	\caption{\Hasse[P_1][]}
	\label{fig:Hasse:wc}
  \end{subfigure}
  \quad~
  \begin{subfigure}{\WIDTH}
	\centering
	\includegraphics[width=\linewidth]{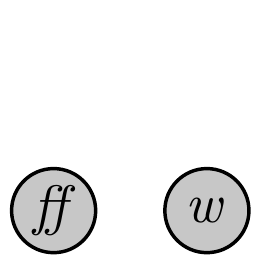}
	\caption{\Hasse[P_2][]}
	\label{fig:Hasse:fw}
  \end{subfigure}
  \quad~
  \begin{subfigure}{\WIDTH}
	\centering
	\includegraphics[width=\linewidth]{Hasse_ff_w_cc}
	\caption{\Hasse[P_3][]}
	\label{fig:Hasse:fwc}
  \end{subfigure}
  \caption{Ordering graph and Hasse diagrams of attacks of \sampleTd: \
	ordering graph \OG, \
	attack $P_1=\poset[\{\ww,\cc\}][\{\ww\prec\cc\}]$, \
	attack $P_2=\poset[\{\ff,\ww\}][\emptyset]$, \
	attack $P_3=\poset[\{\ff,\ww,\cc\}][\{\ww\prec\cc\}]$}
  \label{fig:Hasse}
\end{figure}

We now formalise these concepts, and write $\metr{\T}$ for the unique value $\metr{\ssem{\T}}$ of the dynamic \AT \T, thus mapping \Cref{def:DAT:metric} to the generic notion of metric given in \Cref{def:metric}.
\vspace{-1.5ex}

\begin{definition}
	\label{def:DAT:metric}
	\def\VH{\vphantom{\bigoperAND_{\Hasse}}}
	Let $\operOR,\operAND,\operSAND$ be three associative and commutative
	operators over a set \Vdom:
	we call $\domain=(\Vdom,\operOR,\operAND,\operSAND)$
	a \emph{dynamic attribute domain}.
	Let \T be a well-formed dynamic \AT and $\attrOp$ an attribution on \Vdom.
	The \emph{metric for \T} associated to \domain and $\attrOp$ is given by:
	\begin{align*}
		\metr{\T} &=
			\underbrace{\VH\bigoperOR_{\poset\in\dsem{\T}}}_{\mathlarger\metrOp}
			\underbrace{\VH\;\bigoperAND_{\ccomp\in\Hasse}\;}_{\mathlarger\metrAOp}
			\underbrace{\VH~\bigoperSAND_{a\in\ccomp}~}_{\mathlarger\metrSOp}
			\attr{a}
	\end{align*}
	where \Hasse is the Hasse diagram of attack \poset, and
	$a\in\ccomp$ ranges over the nodes of the connected component \ccomp of \Hasse.
\end{definition}

\begin{example}
	\label{ex:DAT:metric}
	\def\wrap#1{\Bigg(\raisebox{1ex}{$\displaystyle #1$}\Bigg)}
	The semantics of the dynamic \AT from \Cref{ex:running_examples} is
	$\dsem{\sampleTd}= \{
	\poset[\{\ww,\cc\}][\{\ww\prec\cc\}]\,,
	\poset[\{\ff,\ww\}][\emptyset]\}$.
	The Hasse diagrams of these attacks---which resp.\ have one and two
	connected components---are shown in \Cref{fig:Hasse:fw,fig:Hasse:wc}.
	To compute the \emph{min time} metric of \sampleTd consider the attribution
	${\attrOp}=\{{\ff\mapsto3}, {\ww\mapsto15}, {\cc\mapsto1}\}$
	and the dynamic attribute domain $\domain=(\NN,\min,\max,+)$.
	Then the time of the fastest attack for \domain and $\attrOp$ is:
	\begin{align*}
	\metr{\sampleTd}
		&= \bigoperOR_{\poset\in\dsem{\T}}
		   \bigoperAND_{~\ccomp\in\Hasse~}
		   \bigoperSAND_{a\in\ccomp}~\attr{a}\\
		&= \wrap{\bigoperAND_{\ccomp\in\Hasse[\{\ff,\ww\}][\emptyset]}
		   \bigoperSAND_{a\in\ccomp}~\raisebox{-.2ex}{$\attr{a}$}}
		   \mathbin{\text{\raisebox{.3ex}{$\operOR$}}}
		   \wrap{\bigoperAND_{\ccomp\in\Hasse[\{\ww,\cc\}][\ww\prec\cc]}
		   \bigoperSAND_{a\in\ccomp}~\raisebox{-.2ex}{$\attr{a}$}}\\
		&= \big(\attr{\ff}\operAND\attr{\ww}\big)
		   \operOR
		   \big(\attr{\ww}\operSAND\attr{\cc}\big)\\
		&= (3 \maxb 15) \minb (15 + 1) ~=~ 15.
	\end{align*}
	In that computation, attack $\poset[\{\ff,\ww\}][\emptyset]
	\equiv \Hasse[\{\ff,\ww\}][\emptyset]$ has two parallel steps: two
	connected components with one node each---see \Cref{fig:Hasse:fw}---so
	operator $\operAND$ has two operands with one node each:
	$\ccomp=\{\ff\}$ and $\ccomp'=\{\ww\}$.
	In contrast, $\poset[\{\ww,\cc\}][\{\ww\prec\cc\}]
	\equiv \Hasse[\{\ww,\cc\}][\ww\prec\cc]$ has one
	connected component with two nodes---see \Cref{fig:Hasse:wc}---so
	operator $\operAND$ has one operand but $\operSAND$ has two:
	\attr{\ww} and \attr{\cc}.
	Finally, the \emph{min time} of \sampleTd is the ${\operOR}=\min$ of these
	two metrics: the one for \poset[\{\ff,\ww\}][\emptyset].%
	\\ \vphantom{\raisebox{.3ex}{$\Big)$}}%
	Now consider the attributes
	${\attrOp'}=\{{\ff\mapsto42},{\ww\mapsto10},{\cc\mapsto0}\}\!$
	of \emph{min skill} required for each \BAS of \sampleTd.
	Min skill is oblivious of sequential order:
	the skill needed to perform a task is independent of whether
	it must wait for the completion of other tasks.
	So, to compute the \emph{min skill} metric of \sampleTd we use the
	dynamic attribute domain $\domain'=(\NN,\min,\max,\max)$, where the
	operators $\operAND$ and $\operSAND$ are the same.
	This results in:
	\begin{align*}
	\metrOp'(\sampleTd)
		&= \big(\attrOp'(\ff)\operAND'\attrOp'(\ww)\big)
		   \operOR'
		   \big(\attrOp'(\ww)\operSAND\!'\attrOp'(\cc)\big)\\
		&= (42 \maxb 10) \minb (10 \maxb 0) ~=~ 10.
	\end{align*}
\end{example}

\begin{example}
	\label{ex:DAT:metric_vs_Barby}
	Consider the \DAG-structured dynamic \AT from \Cref{fig:wellformed:valid},
	$\T[3] = \AND\big(\SAND(a,b),\SAND(b,c)\big)$, whose ordering graph is
	$\ograph{\T[3]}={a\before b\before c}$ which yields the semantics
	$\dsem{\T[3]}=\{\poset[\{a,b,c\}][\{a\prec b, b\prec c\}]\}$.
	To compute the min attack cost let $\operAND=\operSAND={+}$ and
	$\attrOp=\{{a\mapsto1},{b\mapsto4},{c\mapsto8}\}$ as in the comparison
	with \cite{JKM+15}.
	The Hasse diagram of the poset in \dsem{\T[3]} has one connected component
	with three nodes, so:
	\begin{align*}
	\metr{\T[3]}
		&= \bigoperOR_{\poset\in\dsem{\T[3]}}
		   \bigoperAND_{~\ccomp\in\Hasse~}
		   \bigoperSAND_{a\in\ccomp}~\attr{a}\\
		&=~ \attr{a}\operSAND\attr{b}\operSAND\attr{c} ~=~ 1+4+8 ~=~ 13.
	\end{align*}
\end{example}

Many metrics are like min skill and cost in \Cref{ex:DAT:metric,ex:DAT:metric_vs_Barby}: insensitive to the sequentiality of events.
Therefore, reproducing \Cref{tab:SAT:metric} for dynamic \ATs will introduce a third column for operator $\operSAND$ which resembles the column for $\operAND$.
A main relevant exception is min attack time, where $\operSAND={+}$ because each \BAS in an order-dependency chain must wait for the completion of its predecessor, whereas $\operAND=\max$ yields the time of the slowest parallel part of the attack.

Note also that \emph{the order of execution} of the \BAS in the connected components of an attack \emph{is irrelevant for the computation of a metric}.
This is a direct consequence of the commutativity of the operator $\operSAND$.


\section{Computations for tree-structured DATs}
\label{sec:DAT_trees}

A precondition for our results is that the dynamic \ATs are well-formed as per \Cref{def:wellformed}.
\Cref{alg:is_well_formed} checks this by building the edge relation $\before$ of the ordering graph $\ograph{}=(\BAS,\before)$, and invoking a routine that checks whether \ograph{} has directed cycles.
%
%
\Cref{alg:is_well_formed} terminates after at most \bigO{n^2m} steps (i.e.\ additions of pairs to $\before$), where $n=\card{\BAS}$ and $m$ is the number of \SAND gates.
Ideally one would operate with the transitive reduction of $\before$, computable in less than \bigO{n^{2.5}} \cite{AGU72}.
\vspace{-.5ex}


\begin{algorithm}
	\def\edges{\ensuremath{\mathit{edges}}}
	\KwIn{Dynamic attack tree $\T=(\ATnodes,\typOp,\chOp)$.}
	\KwOut{Whether \T is a well-formed \DAT.}
	\edges\ := $\emptyset$\\
	\ForEach{$\SAND(v_1,\ldots,v_{n+1})\in\ATnodes$}{%
		\For{$i=1$ \KwTo $n$}{%
			\edges\ := $\edges \cup \big(\desc(v_i)\times\desc(v_{i+1})\big)$
		}
	}
	\Return{$\nexists~\text{directed cycle in}~\ograph{}=(\BAS,\edges)$}
	\caption{$\iswellformed(\T)$}
	\label{alg:is_well_formed}
\end{algorithm}

\vspace{-.5ex}
Earlier in \Cref{ex:DAT:metric}, the computation of metrics for dynamic \ATs was illustrated using \Cref{def:DAT:metric}, which is worst-case exponential in the number of nodes.
However and as for \SATs, there is a linear bottom-up algorithm to compute metrics for tree-structured \DATs.
We present a recursive version in \Cref{alg:bottom_up_DAT}, and state its correctness in \Cref{theo:bottom_up_DAT}.

\ifthenelse{\boolean{forCSF}}%
{
To prove the \SAND case of the \namecref{theo:bottom_up_DAT}, operator $\operSAND$ must distribute over $\operOR$ and $\operAND$;
the rest are trivial extensions---to attacks as posets---of the same cases from \Cref{theo:bottom_up_SAT}.
In \cite{arXiv} we give the full proof of \Cref{theo:bottom_up_DAT} by structural induction on \T.
\par
}{
For the case $\SAND(v_1,v_2)$, the proof of \Cref{theo:bottom_up_DAT} (in \Cref{sec:proofs}, \cpageref{theo:bottom_up_DAT:proof}) uses the distributivity of operator $\operSAND$ over $\operOR$ and $\operAND$.
}%
Thus and importantly, besides the tree-structure of the dynamic \AT, the correctness of \Cref{alg:bottom_up_DAT} requires the presence of three semiring algebraic structures: not only $(\Vdom,\operOR,\operAND)$ as in the static case, but also $(\Vdom,\operOR,\operSAND)$ and $(\Vdom,\operAND,\operSAND)$.

\begin{definition}
	\label{def:semiring_dynamic_attribute_domain_OMG_howlong_isthislabel}
	A \emph{semiring dynamic attribute domain} is a dynamic attribute domain
	$\domain=(\Vdom,\operOR,\operAND,\operSAND)$ where
	operator $\operSAND$ distributes over $\operAND$ and $\operOR$,
	and also $\operAND$ distributes over~$\operOR$.
	%
	%
\end{definition}


\begin{algorithm}
	\KwIn{Dynamic attack tree $\T=(\ATnodes,\typOp,\chOp)$,\newline
	      node $v\in\ATnodes$,\newline
	      attribution $\attrOp$,\newline
	      semiring dynamic attr.\ dom.\ ${\domain=(\Vdom,\operOR,\operAND,\operSAND)}$.}
	\KwOut{Metric value $\metr{\T}\in\Vdom$.}
	\BlankLine
	\uIf{$\type{v}=\tOR$}{%
		\Return{$\bigoperOR_{u\in\child{v}}
		         \BUDAT(\T,u,\attrOp,\domain)$}
	} \uElseIf{$\type{v}=\tAND$}{%
		\Return{$\bigoperAND_{u\in\child{v}}
		         \BUDAT(\T,u,\attrOp,\domain)$}
	} \uElseIf{$\type{v}=\tSAND$}{%
		\Return{$\text{\raisebox{.3ex}{%
		        $\bigoperSAND_{u\in\child{v}}
		         \BUDAT(\T,u,\attrOp,\domain)$}}$}
	} \Else(\tcp*[h]{$\type{v}=\tBAS$}) {%
		\Return{\attr{v}}
	}
	\caption{\BUDAT for a tree-structured \DAT \T}
	\label{alg:bottom_up_DAT}
\end{algorithm}

\begin{theorem}
	\label{theo:bottom_up_DAT}
	Let \T be a well-formed tree-structured \DAT,
	$\attrOp$ an attribution on \Vdom,
	and $\domain=(\Vdom,\operOR,\operAND,\operSAND)$ a
	semiring dynamic attribute domain.
	Then $\metr{\T} = \BUDAT(\T,\ATroot,\attrOp,\domain)$.
\end{theorem}

\smallskip

\section{Computations for DAG-structured DATs}
\label{sec:DAT_DAGs}

\Cref{alg:bottom_up_DAT} does not work for dynamic \ATs with a \DAG-structure, for the same reasons exposed for static \ATs in \Cref{sec:SAT_DAGs}.
Neither is it possible to propose algorithms based on standard \BDD theory: even though the structure function of \SATs was reused in \Cref{def:DAT:semantics}, the computation of metrics for \DATs intrinsically needs a notion of order among their \BAS, that is not present in standard \BDD-based data types.

As discussed in \Cref{sec:intro}, some earlier general approaches do exist to compute metrics on \DAG-structured \DATs \cite{KRS15,AGKS15}.
However, these often overshoot in terms of computation complexity.
For static \ATs and from a procedural (rather than semantic) angle, \cite{KW18} proposes a more efficient, ingenious approach that computes and then corrects a metric value by traversing the \AT bottom-up multiple times.
It may be possible to extend this algorithm to consider \SAND gates as well \cite{WAFP19}.

Alternatively, \Cref{def:DAT:metric} of metric for \DATs could be encoded into a na\"ive algorithm.
This would enumerate all posets from \dsem{\T}, and compute the metric value \metr{\T} by means of three nested loops that traverse all these Hasse diagrams.
We do not expect such approach to be computationally efficient.

Instead and as in the static case, we expect that \BDD encodings of the \DAT offer better solutions.
This requires \BDD structures that are somehow sensitive to variable orderings.
In that sense, the so-called sequential-\BDDs recently presented for dynamic fault trees seem promising \cite{YW20}.
A first challenge would be to extend them to attributes other than failure (viz.\ attack) probability.
Harder to tackle is the combinatorial explosion, that stems for the different possible orderings of \BAS descendants of \SAND gates.

In view of these considerations, we regard the algorithmic analysis for \DAG-structured dynamic attack trees as an important open problem for future research. 


\section{Conclusions \texorpdfstring{\href{https://youtu.be/CjHX0PRFq2Q}{\color{white}.}\!\!\!}{}}
\label{sec:conclu}

This paper presents algorithms to compute quantitative security metrics on attack trees.
Our approach is formal: we classify \AT models based on their structure and components, and then for each class we:
\begin{enumerate*}[label=(\arabic*)]
\item	revise and consolidate its semantics in line with the literature,
\item	define metrics generically on these semantics,
\item	present algorithms to compute them, and
\item	show the correctness of our algorithms,
		as well as their optimality in terms of computational complexity.
\end{enumerate*}

\Cref{alg:bottom_up_BDD} is a prominent result: it computes metrics efficiently in \DAG-structured static \ATs, from a given semiring attribute domain with neutral elements.
Another key contribution is the poset semantics defined for \DATs in \Cref{sec:DAT}:
it lifts the concepts used for \SATs in a simple manner, which nevertheless allows computations on \DAG-structured models.

We noted that our \DAT semantics rules out some interleavings in the execution of \SAND gates, e.g.\ $(b,a,c)$ for ${\SAND\big(a,\AND(b,c)\big)}$.
To allow such sequences it is necessary to use formulae---rather than individual \BAS---as nodes of an ordering graph.
For the \DAT above, this would yield the edge $a\before(b\land c)$, which allows $(b,a,c)$ because the formulae in that sequence are satisfied in the order ``first $a$, then $b\land c$.''

Interestingly, such formula-based ordering graphs preserve the coherence of our semantics, because \Cref{prop:coherent_DATs} does not depend on the objects represented by the nodes:
it just requires that traversing edges on the ordering graph represents valid execution orders of the children of \SAND gates.
Therefore, such ordering graphs are a promising research direction.

Further lines for future work also include:
developing efficient algorithms to compute metrics
on \DAG-structured dynamic \ATs;
extending our \AT syntax to include \emph{sequential-\OR gates}
\cite{KRS15,HKKS16};
extending our metrics to consider attacker profiles; and
combining tree and \DAG structures in a clever way,
e.g.\ computing values linearly for the tree components,
and plugging these into the rest of the (\DAG) structure.

\subsection*{\bfseries Related work}
\label{sec:conclu:related_work}


\begin{table*}
	\centering
	
	\vspace{1ex}
	\caption{%
	  Algorithms for metrics on different \AT classes%
	  ~~\textcolor{black!70}{(replica of \Cref{tab:all_algos_intro})}
	}
	\label{tab:all_algos_conclu}
	\vspace{2ex}
	\parbox{.73\linewidth}{
		\def\bfacro#1{\acronym{\bfseries{#1}}}%
		\bfacro{bu}:      bottom-up on the \AT structure.
		\bfacro{aph}:     acyclic phase-type (time distribution).
		\bfacro{bdd}:     binary decision diagram.
		\bfacro{mtbdd}:   multi-terminal \acronym{bdd}.
		$\pmb{\pazocal{C}}$-\bfacro{bu}:
		                  repeated \acronym{bu}, identifying clones.
		\bfacro{dpll}:    \acronym{dppl sat}-solving in the \AT formula.
		\bfacro{pta}:     priced time automata (semantics).
		\bfacro{i/o-imc}: input/output interactive Markov chains (semantics).
	}
	\vspace{-2ex}
\end{table*}

Surveys on attack trees are \cite{KPS14,WAFP19}:
the latter covers \AT analysis via formal methods, from which we are close to quantitative model checking---cf.\ simulation studies such as \cite{DMCR06,WAN18}.
Concrete case studies have been reported in \cite{FGK+16}.


\Cref{tab:all_algos_conclu}
condenses literature references on quantitative analyses of \ATs, classified by the structure and (dynamic) gates of the \ATs where they are applicable.
For each metric and \AT class, in this \namecref{tab:all_algos_conclu} we cite the earliest relevant contributions that include concrete computation procedures.

Works \cite{BLP+06,JW08} are among the first to model and compute the cost and probability of attacks: their algorithms have \EXPTIME complexity regardless of the \AT structure.
In \cite{KRS15,KSR+18} an attack is moreover characterised by the time it takes.
This allows for richer Pareto analyses but introduces one clock per \BAS in the Priced Time Automata semantics: algorithms have thus \EXPTIME \& \PSPACE complexity \cite{AD94,BLR05}.
The current work improves these bounds via specialised procedures tailored for the specific \AT class, e.g.\ \Cref{alg:bottom_up_SAT,alg:bottom_up_DAT} resp.\ for tree-structured \SATs and \DATs have \LINTIME complexity.

Indeed, all algorithms specialised on tree-structured \ATs implement a bottom-up traversal on its syntactic structure: we denote these \acronym{bu} in \Cref{tab:all_algos_conclu}.
Pareto analyses are polynomial, where the exponent is the number of parameters being optimised.
Most works are on static \ATs, with the relevant exception of \cite{AHPS14,JKM+15} which include sequential \AND gates.

For \DAG-structured static \ATs the algorithmic spectrum is broader, owing to the NP-hardness of the problem (see \Cref{sec:SAT_DAGs:complexity}).
Such algorithms range from classical \BDD encodings for probabilities, and extensions to multi-terminal \BDDs, to logic-based semantics that exploit \acronym{dpll}, including an encoding of \SATs as generalised stochastic Petri nets.
A prominent contribution is \cite[Alg.~1]{KW18}: after computing so-called optional and necessary clones, its computations are exponential on the number of shared \BAS (only).

In contrast, the computation of security metrics for dynamic attack trees is more recent than for \SATs: here we find open problems in the literature, indicated in four cells of \Cref{tab:all_algos_conclu}.
For tree-structured \DATs and to the best of our knowledge, no work addresses directly the computation of Pareto frontiers.
For this, our \Cref{alg:bottom_up_DAT} ($\mathtt{BU_{DAT}}$) could be embedded in the static setting of \cite{AN15,HAF+09}: the gist would be to carry around pairs of values instead of only one, removing dominated solutions at each step.
As for $k$-top metric algorithms, our \Cref{alg:bottom_up_DAT} could be extended with priority lists updated during the tree traversal.

We thus propose to tackle two open problems on tree-structured \DATs, by simple combinations or extensions of other methods (from the literature or introduced in this work).
In contrast, the open problems for \DAG-structured \DATs are less easy to overcome.
To compute attack probability, \cite{AGKS15} encodes these attack trees as a variant of Markov chains.
For other metrics, \cite{KRS15} encodes the \AT as a network of \acronym{pta} and solves the resulting cost-optimal reachability problem.
As earlier stated, these very powerful and general approaches are in detriment of computational efficiency.

A recent related approach encodes dynamic fault trees as so-called sequential-\BDDs, to compute the probability of system failure \cite{YW20}.
Such safety-oriented works do not map directly to security analysis such as \AT metrics, because:
\begin{enumerate*}
\item	they can compute probability---and possibly parallel time---only;
\item	the dynamic gates are not the same than those in dynamic \ATs;
\item	the standard logical gates are interpreted differently.
\end{enumerate*}
Still, it might be feasible to adapt \cite{YW20} to compute \AT metrics, e.g.\ to compare it against the algorithms here presented.
Probably the main detriment is that sequential-\BDDs expand sequence dependencies of every pair of events, adding a combinatorial blow-up on top of the already exponential explosion incurred by \BDD representations of \DAGs.
This leads us to believe that even the \EXPTIME complexity of our \Cref{alg:bottom_up_BDD} should be more time-efficient.

\bibliographystyle{IEEEtran}
\bibliography{\jobname}

\ifthenelse{\boolean{forCSF}}%
{
}{
	\newpage\clearpage
	\appendices
\section{Proofs of results from the paper}
\label[appendix]{sec:proofs}  

\begingroup

\allowdisplaybreaks

\def\REF#1{\textit{\ref{#1})}}
\def\REFS#1#2{\textit{\ref{#1})--\ref{#2})}}
\def\hop{\\[.5ex]}
\def\by#1{&&\hspace*{-\linewidth}\text{\smaller\color{black!70}~by #1}}


\begin{replemma}{lemma:ssem}
	Consider a \SAT with nodes $a\in\BAS, v_1, v_2\in\ATnodes$,
	that has a proper tree structure. Then:
	\begin{enumerate}
	\item	$\ssem{a} = \{ \{a\}\}$;
			\label{lemma:ssem:BAS}
	\item	$\ssem{\OR(v_1,v_2)} = \ssem{v_1} \cup \ssem{v_2}$;
			\label{lemma:ssem:OR}
	\item	$\ssem{\AND(v_1,v_2)} = \{ \attack_1\cup \attack_2 \mid
				\attack_1\in\ssem{v_1} \land \attack_2\in\ssem{v_2} \}$;
			\label{lemma:ssem:AND}
	\item	In cases \REF{lemma:ssem:OR} and \REF{lemma:ssem:AND}
			the \ssem{v_i} are disjoint, and in case \REF{lemma:ssem:AND}
			moreover the $A_i$ are pairwise disjoint.
			\label{lemma:ssem:disjoint}
	\end{enumerate}
\end{replemma}
\begin{proof}	\label{lemma:ssem:proof}
	The semantics function in \Cref{def:SAT:semantics}, $\ssem{\cdot}\from\allATs\to\allSuites$, takes as argument a static attack tree $\T=(\ATnodes,\typOp,\chOp)$.
	This \namecref{lemma:ssem} presents (abridgedly) and overloaded function $\ssem{\cdot}\from\ATnodes\to\allSuites$, for which we want to prove that $\ssem{\ATroot}=\ssem{\T}$.
	We do this by structural induction in \T, considering that each case~\REFS{lemma:ssem:BAS}{lemma:ssem:AND} gives semantics to the root of the corresponding \SAT.
	But first note that case~\REF{lemma:ssem:disjoint} is straightforward given the tree-structure of \T: the \BAS descendants of the children $v_1$ and $v_2$ form a partition of the set $\BAS_{\T}$, so no \BAS that appears in \ssem{v_1} can appear in \ssem{v_2} and vice versa.
	Now we prove cases~\REFS{lemma:ssem:BAS}{lemma:ssem:AND}:
	
	\begin{enumerate}[topsep=.3ex,parsep=0pt,itemsep=.3ex]
	
	\item[\REF{lemma:ssem:BAS}]
	We have $\ATroot=a\in\BAS$, so by \Cref{def:AT:syntax}:
	\mbox{$\ATnodes=\{a\}$}, and thus $\allSuites_{\T}=\{\emptyset,\{a\}\}$.
	It follows that:
	\mbox{$\sfunT(a,\attack)=1$} iff \mbox{$\attack=\{a\}$}
	$~\therefore~ \ssem{\T}=\{\{a\}\}=\ssem{\ATroot}$.
	
	\item[\REF{lemma:ssem:OR}]
	For $i=1,2$ let \T[i] be the \SAT s.t.\ $\ATroot[{\T[i]}]=v_i$ in case~\REF{lemma:ssem:OR}, then by IH $\ssem{\T[i]}=\ssem{v_i}$.
	Due to case~\REF{lemma:ssem:disjoint} it follows that an attack $\attack\in\ssem{v_1}\cup\ssem{v_2}=\ssem{v_1}\uplus\ssem{v_2}$ iff $\sfunT[{\T[i]}](\attack)=\top$ and \attack is minimal for one of the \T[i], where $\uplus$ denotes disjoint union.
	W.l.o.g.\ let this hold for \T[1].
	Since $\type{\ATroot}=\tOR$ then by \Cref{def:SAT:sfun,def:SAT:semantics} this happens iff $\sfunT(\attack)=\top$ and \attack is minimal for \T[1], the latter by our current hypothesis.
	But in turn this happens iff \attack is also minimal for \T, because otherwise $\exists\,\attack[B]\subsetneq\attack\,.\,\sfunT(\attack[B])=\top$, and since no element of \attack can be in the \BAS descendants of $v_2$, this would mean that $\sfunT[{\T[1]}](\attack[B])=\top$, which contradicts the hypothesis that $\attack\in\ssem{v_1}$.
	Therefore, $\attack\in\ssem{v_1}\cup\ssem{v_2}$ iff $\sfunT(\attack)=\top$ and \attack is minimal for \T, i.e.\ iff $\attack\in\ssem{\T}$.

	\item[\REF{lemma:ssem:AND}]
	Let \T[1],\T[2] be as before and consider an attack $\attack\in\ssem{\T}$.
	By case~\REF{lemma:ssem:disjoint} we can partition $\attack=\attack_1\uplus\attack_2$.
	Since $\type{\ATroot}=\tAND$ and by \Cref{def:SAT:sfun,def:SAT:semantics}, $\attack\in\ssem{\T}$ iff $\sfunT[{\T[i]}](\attack_i)$ for both children and \attack is minimal for \T.
	Now, if $\attack_1$ were not minimal for \T[1], then $\exists\,\attack[B]_1\subsetneq\attack_1\,.\,\sfunT[{\T[1]}](\attack[B]_1)=\top$.
	But then $\attack'=\attack[B]_1\uplus\attack_2\subsetneq\attack$ is a sucessful attack for \T, which contradicts the hypothesis that $\attack\in\ssem{\T}$.
	By an analogous argument with $\attack_2$ we get that $\attack\in\ssem{\T}$ iff $\sfunT[{\T[i]}](\attack_i)$ and $\attack_i$ is minimal for \T[i], i.e.\ iff $\attack\in\{\attack_1\cup\attack_2\mid\attack_i\in\ssem{v_i}\}$.

	\end{enumerate}
\end{proof}


\begin{replemma}{lemma:dsem}
	Consider a well-formed \DAT with nodes ${a\in\BAS}$, ${v_1,v_2\in\ATnodes}$,
	that has a proper tree structure. Then:
	\begin{enumerate}
	\item	$\dsem{a} = \{ \poset[\{a\}][\emptyset] \}$;%
			\label{lemma:dsem:BAS}
	\item	$\dsem{\OR(v_1,v_2)} = \dsem{v_1} \cup \dsem{v_2}$;%
			\label{lemma:dsem:OR}
	\item	$\dsem{\AND(v_1,v_2)} = \mbox{$\left\{
				\poset[\attack_1{\cup}\attack_2][{\prec_1}{\cup}{\prec_2}]
				\,|\, \poset[\attack_i][\prec_i]\in\dsem{v_i}
			\right\}$}$;%
			\label{lemma:dsem:AND}
	\item	$\dsem{\SAND(v_1,v_2)} = \{
				\poset[\attack_1\cup\attack_2 \,]%
				      [~{\prec_1}\cup{\prec_2}\cup{\attack_1\times\attack_2}]
				\cdots$ \\
				\hspace*{\stretch{1}} $\cdots
					\mid \mbox{$\poset[\attack_i][\prec_i]\in\dsem{v_i}$}
			\}$;%
			\label{lemma:dsem:SAND}
	\item	In cases \REFS{lemma:dsem:OR}{lemma:dsem:SAND} above the
			\dsem{v_i} are disjoint, and in cases \REF{lemma:dsem:AND} and
			\REF{lemma:dsem:SAND} moreover the $A_i$ are pairwise disjoint.
			\label{lemma:dsem:disjoint}
	\end{enumerate}
\end{replemma}
\begin{proof}	\label{lemma:dsem:proof}
	Cases~\REF{lemma:dsem:BAS}, \REF{lemma:dsem:OR}, and \REF{lemma:dsem:disjoint}, are a trivial extension---to attacks as posets---of the same cases from \Cref{lemma:ssem}.
	To prove cases~\REF{lemma:dsem:AND} and \REF{lemma:dsem:SAND} we proceed as in \Cref{lemma:ssem} by structural induction on \T, the (dynamic) \AT whose root is the node under consideration.
	Therefore, for ${i=1,2}$ we have \DATs \T[i] s.t.\ $\ATroot[{\T[i]}]=v_i$ and (by IH) $\dsem{\T[i]}=\dsem{v_i}$.

	\begin{enumerate}[topsep=.3ex,parsep=0pt,itemsep=.3ex]
	
	\item[\REF{lemma:dsem:AND}]
	We must show that ${\poset\in\dsem{\AND(v_1,v_2)}}$ iff there exist $\poset[\attack_i][\prec_i]\in\dsem{v_i}$ s.t.\ ${A=A_1\cup A_2}$ and ${{\prec}={\prec_1}\cup{\prec_2}}$.
	First note that \Cref{def:DAT:attack} of (minimal) attack for a dynamic \AT \T uses the structure function of its corresponding static \AT $\T'$.
	This means that the sets $\attack_i\subset\BAS_{\T[i]}$ of the poset-attacks $\poset[\attack_i][\prec_i]\in\dsem{\T[i]}$ are under the same restrictions than in the static case.
	Therefore we can reduce to case~\REF{lemma:ssem:AND} from \Cref{lemma:ssem} to ensure that there must indeed exist minimal set-attacks $A_i\in\ssem{\T[i]'}$ s.t.\ $A=A_1\cup A_2\in\ssem{\T'}$.
	Now consider the partial orders ${\prec_i}$\,, which by \Cref{def:wellformed,def:DAT:attack} are minimal relations that respect the order imposed by the \SAND gates of \T[i].
	Since $\T=\AND(v_1,v_2)$ has the same restrictions on the order of \BAS than \T[1] and \T[2] together, and since these two \DATs have disjoint \BAS sets, we get that ${{\prec}={\prec_1}\cup{\prec_2}}$.

	\item[\REF{lemma:dsem:SAND}]
	This proof follows the same idea than for case \REF{lemma:dsem:AND} above with one addendum: here $\T=\SAND(v_1,v_2)$ has more restrictions on the order of \BAS than \T[1] and \T[2].
	First, since \T[1] and \T[2] are subtrees of \T, then all their restriction on \BAS elements apply to \T, so ${\prec_1}\cup{\prec_2}\subseteq{\prec}$.
	Second, on top of these restrictions and by \Cref{def:wellformed,def:DAT:attack}, the ordering graph $\ograph{}=(\BAS,\before)$ of \T imposes that every \BAS from \T[1] comes before any \BAS from \T[2].
	Morevoer, the partial order of a poset-attack $\poset\in\dsem{\T}$ is a restriction of the edge relation $\before$ to \attack.
	Since $\attack=\attack_1\cup\attack_2$ for some $\attack_i\in\ssem{\T[i]'}$, then these extra restrictions are precisely of the form $a_1\prec a_2$ s.t.\ $a_i\in\attack_i$, i.e.\ $\attack_1\times\attack_2$.
	In other words, ${{\prec}={\prec_1}\cup{\prec_2}\cup{\attack_1\times\attack_2}}$.
	\end{enumerate}

\end{proof}


\begin{reptheorem}{theo:bottom_up_SAT}
	Let \T be a static \AT with tree structure,
	$\attrOp$ an attribution on \Vdom,
	and $\domain=(\Vdom,\operOR,\operAND)$ a semiring attribute domain.
	Then $\metr{\T} = \BUSAT(\T,\ATroot,\attrOp,\domain)$.
\end{reptheorem}
\begin{proof}	\label{theo:bottom_up_SAT:proof}
	Let $\BU(v)=\BUSAT(\T,v,\attrOp,\domain)$ for any node $v$ of \T,
	then we want to prove that $\metr{\T}=\BU(\ATroot)$.
	W.l.o.g.\ we consider binary trees (to make use of \Cref{lemma:ssem}),
	and proceed by structural induction in \T.
	Since this is a static \AT there are only three possibilities for \ATroot:
	\begin{enumerate}[topsep=0pt,itemsep=.2ex,leftmargin=1.3em]

	\item	$\ATroot=a\in\BAS$:
	\\ then by \Cref{lemma:ssem}.\ref{lemma:ssem:BAS}) 
	$\ssem{\T}=\ssem{\ATroot}=\{\{a\}\}$, so by \Cref{def:SAT:metric}:
	\begin{align*}
	\metr{\T}
		&= \bigoperOR_{\attack\in\ssem{\T}}
		   \bigoperAND_{b\in\attack} \attr{b}\\
		&= \attr{a}\\
		&= \BU(\ATroot).
	\end{align*}

	\item	$\ATroot=\OR(v_1,v_2)$:
	\begin{align*}
	\metr{\T}
		&= \bigoperOR_{\attack\in\ssem{\T}} \metrA{\attack}
		\hop
		\by{\Cref{lemma:ssem}.\ref{lemma:ssem:OR})}
		\hop
		&= \bigoperOR_{\attack\in\ssem{v_1}\cup\ssem{v_2}} \metrA{\attack}
		\hop
		\by{assoc.\ and conm.\ of $\operOR$}
		\hop
		&= \bigoperOR_{\attack\in\ssem{v_1}}
		   \bigoperOR_{\attack\in\ssem{v_2}\setminus\ssem{v_1}}
		   \metrA{\attack}
		\hop
		\by{\Cref{lemma:ssem}.\ref{lemma:ssem:disjoint})}
		\hop
		&= \bigoperOR_{\attack\in\ssem{v_1}}
		   \bigoperOR_{\attack\in\ssem{v_2}}
		   \metrA{\attack}
		\hop
		\by{distr.\ of $\operAND$ over $\operOR$}
		\hop
		&= \left( \bigoperOR_{\attack\in\ssem{v_1}} \metrA{\attack} \right)
		   \operOR
		   \left( \bigoperOR_{\attack\in\ssem{v_2}} \metrA{\attack} \right)
		\hop
		\by{\Cref{def:SAT:metric} and IH}
		\hop
		&= \BU(v_1) \operOR \BU(v_2)
		\hop
		\by{\Cref{alg:bottom_up_SAT}}
		\hop
		&= \BU(\ATroot).
	\end{align*}

	\item	$\ATroot=\AND(v_1,v_2)$:
	\begin{align*}
	\metr{\T}
		&= \bigoperOR_{\attack\in\ssem{\T}} \metr{\attack}
		\hop
		\by{Hyp. and \Cref{def:SAT:metric}}
		\hop
		&= \bigoperOR_{\attack\in\ssem{\AND(v_1,v_2)}}~
		   \bigoperAND_{a\in\attack} \attr{a}
		\hop
		\by{\nameCrefs{lemma:ssem} \labelcref{lemma:ssem}.\ref{lemma:ssem:AND})
		    and \labelcref{lemma:ssem}.\ref{lemma:ssem:disjoint})}
		\hop
		&= \bigoperOR_{\attack_1\uplus\attack_2\in\ssem{v_1}\uplus\ssem{v_2}}~
		   \bigoperAND_{a\in\attack_1\uplus\attack_2} \attr{a}
		\hop
		\by{asoc.\ and conm.\ of $\operOR$,
		    and $\ssem{v_1}\cap\ssem{v_2}=\emptyset$}
		\hop
		&= \bigoperOR_{\attack_1\in\ssem{v_1}} \left(
		   \bigoperOR_{\attack_2\in\ssem{v_2}} \Bigg(
		       \bigoperAND_{a\in\attack_1\uplus\attack_2} \attr{a}
		   \Bigg)\right)
		\hop
		\by{asoc.\ and conm.\ of $\operAND$,
		    and $\attack_1\cap\attack_2=\emptyset$}
		\hop
		&= \bigoperOR_{\attack_1\in\ssem{v_1}} \mathlarger{\Bigg(}
		   \bigoperOR_{\attack_2\in\ssem{v_2}} ~~\cdots\\
		   &\qquad\cdots
		   \Bigg( 
		       \bigg(\bigoperAND_{a_1\in\attack_1}\attr{a_1}\bigg)
		       \operAND
		       \bigg(\bigoperAND_{a_2\in\attack_2}\attr{a_2}\bigg)
		   \Bigg)\mathlarger{\Bigg)}
		\hop
		\by{distrib.\ of $\operAND$ over $\operOR$}
		\hop
		&= \bigoperOR_{\attack_1\in\ssem{v_1}} \mathlarger{\Bigg(}
		   \bigg(\bigoperAND_{a_1\in\attack_1}\attr{a_1}\bigg)
		   \operAND ~~\cdots\\
		   &\qquad\cdots
		   \Bigg(
		       \bigoperOR_{\attack_2\in\ssem{v_2}} 
		       \bigoperAND_{a_2\in\attack_2}\attr{a_2}
		   \Bigg)\mathlarger{\Bigg)}
		\hop
		\by{distrib.\ of $\operAND$ over $\operOR$}
		\hop
		&= \left( \bigoperOR_{\attack_1\in\ssem{v_1}}
		          \bigoperAND_{a_1\in\attack_1}\attr{a_1} \right)
		   \operAND ~~\cdots\\
		   &\qquad\cdots
		   \left( \bigoperOR_{\attack_2\in\ssem{v_2}}
		          \bigoperAND_{a_2\in\attack_2}\attr{a_2} \right)
		\hop
		\by{\Cref{def:SAT:metric} and IH}
		\hop
		&= \BU(v_1)\operAND\BU(v_2)
		\hop
		\by{\Cref{alg:bottom_up_SAT}}
		\hop
		&= \BU(\ATroot),
	\end{align*}
	where $\uplus$ denotes disjoint union.
	\end{enumerate}

\end{proof}


\begin{reptheorem}{theo:NP_hard}
	The problem of computing the smallest minimal attack
	of a DAG-structured static AT is NP-hard.
\end{reptheorem}
\begin{proof}	\label{theo:NP_hard:proof}
	\def\minSAT{\textit{minSAT}\xspace}
	\def\CNFSAT{\textit{CNFSAT}\xspace}
	\def\posvar#1{\raisebox{.95ex}{\rotatebox{180}{\ensuremath{#1}}}}
	\def\posphi{\ensuremath{\widehat{\varphi}}\xspace}
	\def\minmap{\operatorname{\mathit{g}}}
	A static attack tree \T is equivalent to a logical formula whose atoms are the elements of \BAS: denote this formula \LT and note that none of its atoms appears negated.
	The problem of finding the smallest minimal attack in \T can thus be reformulated as finding the smallest $A\subseteq\BAS$ whose elements must evaluate to $\top$ to satisfy \LT.
	Denote this problem \minSAT.
	We now reduce \CNFSAT, the satisfiability problem for arbitrary logic formulae in conjunctive normal form, to solving \minSAT.
	Let $\varphi=\bigwedge_{i=1}^n c_i=\bigwedge_{i=1}^n\bigvee_{j=1}^{N_i}\ell_i^j$ be one such arbitrary formula with atoms $\mathrm{A}=\{a_k\}_{k=1}^m$.
	Define $\pos(\ell)\doteq\posvar{a}$ if the literal $\ell=\neg a$, and $\pos(\ell)\doteq a$ otherwise, where \posvar{a} is a fresh non-negated (``positive'') atom.
	Now let $\hat{c}_i\doteq\bigvee_{j=1}^{N_i}\pos(\ell_i^j)$ and $\hat{a}\doteq(a\lor\posvar{a})$, and consider the formula $\posphi=\bigwedge_{i=1}^n\hat{c}_i\,\bigwedge_{k=1}^m\hat{a}_k$.
	Since no atom of \posphi is negated, by \minSAT we can find some $\minmap\from\mathrm{A}\cup\posvar{\mathrm{A}}\to\BB$ that satisfies \posphi, mapping to $\top$ a \emph{minimum amount of atoms} from $\mathrm{A}\cup\posvar{\mathrm{A}}$.
	Now consider the second part of the conjunction in \posphi: satisfying $\bigwedge_{k=1}^m\hat{a}_k$ requires, minimally, mapping $m$ atoms to $\top$, e.g.\ all the $\mathrm{A}$, or all the $\posvar{\mathrm{A}}$.
	But then:
	\begin{itemize}[topsep=0pt,parsep=.5ex,leftmargin=1em]
	\item	if $\minmap$ maps \emph{exactly $m$ atoms} to $\top$, then for every $\hat{a}_k$ it mapped either $a_k$ or $\posvar{a}_k$ to $\top$\ \;$\therefore$\;\ $\varphi$ is satisfiable;
	\item	else $\exists~\hat{c}_i,a_k$ s.t.\ $a_k\in\hat{c}_i$ and $\posvar{a}_k\in\hat{c}_i$ and $\minmap(a_k)=\minmap(\posvar{a}_k)=\top$\ \;$\therefore$\;\ $\varphi$ is unsatisfiable.
	\end{itemize}

\end{proof}

\begingroup

\def\root{\ensuremath{\BDDroot[{\bddT}]}\xspace}
\def\IC#1{\ensuremath{\mathrm{IC}_{#1}}\xspace}

\begin{reptheorem}{theo:bottom_up_BDD}
	Let \T be a static \AT, \bddT its \BDD encoding over \poset[\BAS][{<}],
	$\attrOp$ an attribution on \Vdom,
	and $\ndomain=(\Vdom,\operOR,\operAND,\ntOR,\ntAND)$ an attribute domain
	with neutral elements resp.\ for $\operOR$ and $\operAND$.
	Then $\metr{\T} = \BUBDD(\bddT,\root,\attrOp,\ndomain)$.
\end{reptheorem}
\begin{proof}	\label{theo:bottom_up_BDD:proof}
	We use induction on the number of leaves of the \AT, which is the number of nodes of \bddT and therefore its number of levels.
	In particular, the inductive step exploits the fact that the leaf $\oldtop$ in \bddT (labelled with $\top$) cannot be the $\low$ child of a nonterminal node $w\in\BDDnodesN$.
	Intuitively, this is because by visiting $\low(w)$ in the traversal of \bddT, \Cref{alg:bottom_up_BDD} considers \emph{the exclusion} of $\BDDlab(w)\in\BAS$ from the current attack under consideration.
	Static \ATs are coherent, so excluding a \BAS cannot be the reason that makes an attack succeed.
	Therefore, taking the $\low$ child of $w$ cannot lead to $\oldtop$.
	An analogous reasoning entails that the leaf $\oldbot$ cannot be a $\high$ child.
	\begin{itemize}[label=\textbullet,topsep=0pt,itemsep=.2ex,leftmargin=1em]
	\item	In the base case
	$\BAS=\{a\}$, so the \BDD has a single nonterminal node labelled with $a$, whose $\low$ child is $\oldbot$ and $\high$ child is $\oldtop$.
	Then by \Cref{lemma:ssem}.\ref{lemma:ssem:BAS}) we get that $\ssem{\T}=\ssem{\ATroot}=\{\{a\}\}$, so by \Cref{def:SAT:metric}:
$\metr{\T}=\attr{a}=\ntOR\operOR(\ntAND\operAND\attr{a})=\BUBDD(\bddT,\root,\attrOp,\ndomain)$.
	\item	Assume by IH
	that the statement holds for any \SAT with $\BAS=\{a_1,\ldots,a_n\}$ and consider \T with basic attack steps ${\BAS\uplus\{a_0\}}$.
	W.l.o.g.\ let $a_0<a_i$ for $0<i\leqslant n$.
	Then ${\BDDlab(\root)=a_0}$ and since \root is nonterminal we get $\BUBDD(\bddT,\root,\attrOp,\ndomain)=\IC{\overbar{a_0}}\operOR(\IC{a_0}\operAND\attr{a_0})$.
	\IC{\overbar{a_0}} stands for the inductive case whose (sub-)\,\BDD has $\low(\root)\in\BDDnodes$ as root node: call this \BDD \bddT[L].
	Proceed analogously for the inductive case \IC{a_0} with \BDD \bddT[H].
	Then both \BDDs \bddT[L] and \bddT[H] represents static \ATs with $\BAS=\{a_i\}_{i=1}^n$.
	Call these \T[L] and \T[H], then by IH: $\metr{\T[L]}={\BUBDD(\bddT,\low(\root),\attrOp,\ndomain)}={\IC{\overbar{a_0}}}$, and analogously for \T[H] and $\high(\root)$.
	Recall now that \bddT represents the structure function of \T by exploiting its Shannon expansion.
	That means that \bddT[L] represents a sub-structure function of \T
in the case that $a_0$ is mapped to $\bot$, i.e.\ when this basic attack step does not occur.
	Therefore, \metr{\T[L]} is the metric that considers all attacks (from \T) that do not require $a_0$.
	Oppositely, \metr{\T[H]} is the metric for all attacks whose success requires that $a_0$ takes place.
	Since any attack in \ssem{\T} either contains $a_0$ or not, the metric \metr{\T} could be computed as the disjunction of \metr{\T[L]} and \metr{\T[H]}.
	However, $a_0\not\in\BAS_{\T[H]}$, so \metr{\T[H]} is not counting \attr{a_0}.
	But the attacks in \T represented by \T[H] require $a_0$ to succeed, so we must include---via the conjunction operator---its attribute to the metric.
	This step requires the distributivity of $\operAND$ over $\operOR$ (and over itself): this allows to embed or ``push-in'' the attribute \attr{a_0} via $\operAND$ in the recursive computation \IC{a_0}.
	In sum: $\metr{\T}={\metr{\T[L]}\operOR(\metr{\T[H]}\operAND\attr{a_0})}={\IC{\overbar{a_0}}\operOR(\IC{a_0}\operAND\attr{a_0})}={\BUBDD(\bddT,\root,\attrOp,\ndomain)}$.
	\end{itemize}

\end{proof}
\endgroup  

\begin{reptheorem}{theo:bottom_up_DAT}
	Let \T be a well-formed tree-structured \DAT,
	$\attrOp$ an attribution on \Vdom,
	and $\domain=(\Vdom,\operOR,\operAND,\operSAND)$ a
	semiring dynamic attribute domain.
	Then $\metr{\T} = \BUDAT(\T,\ATroot,\attrOp,\domain)$.
\end{reptheorem}
\begin{proof}	\label{theo:bottom_up_DAT:proof}
	Let $\BU(v)=\BUDAT(\T,v,\attrOp,\domain)$ for any node $v$ of \T.
	W.l.o.g.\ we consider binary trees (to use \Cref{lemma:dsem}),
	and proceed by structural induction. 
	Cases $\type{\ATroot}\in\{\tBAS,\tOR,\tAND\}$ are a trivial extension---to
	attacks as posets via \Cref{lemma:dsem}---of the same cases from
	\Cref{theo:bottom_up_SAT}.
	Let us then prove the case ${\BU(\ATroot)} = {\BU(\SAND(v_1,v_2))} = {\BU(v_1)\operSAND\BU(v_2)}$, where by IH each call expands to ${\bigoperOR_{\poset[\attack_i][\prec_i]\in\dsem{v_i}} \bigoperAND_{\ccomp_i\in\Hasse[\attack_i][\prec_i]} \bigoperSAND_{a_i\in\ccomp_i} \attr{a_i}}$:
	\begin{flalign*}
	\BU(\ATroot)
		&= \left(
		   \bigoperOR_{\poset[\attack_1][\prec_1]\in\dsem{v_1}}
		   \bigoperAND_{\ccomp_1\in\Hasse[\attack_1][\prec_1]}
		   \bigoperSAND_{a_1\in\ccomp_1} \attr{a_1}
		   \right)
		   \operSAND \cdots\\
		   &\cdots
		   \left(
		   \bigoperOR_{\poset[\attack_2][\prec_2]\in\dsem{v_2}}
		   \bigoperAND_{\ccomp_2\in\Hasse[\attack_2][\prec_2]}
		   \bigoperSAND_{a_2\in\ccomp_2} \attr{a_2}
		   \right)
		\hop
		\by{distr.\ of $\operSAND$ on $\operAND$ and on $\operOR$}
		\hop
		&= \bigoperOR_{\poset[\attack_1][\prec_1]\in\dsem{v_1}}
		   \left(
		   \bigoperAND_{\ccomp_1\in\Hasse[\attack_1][\prec_1]}
		   \Bigg(
		   \bigoperSAND_{a_1\in\ccomp_1} \attr{a_1}
		   \right.
		   \operSAND \cdots\\
		   &\cdots
		   \left.
		   \bigg(
		   \bigoperOR_{\poset[\attack_2][\prec_2]\in\dsem{v_2}}
		   \bigoperAND_{\ccomp_2\in\Hasse[\attack_2][\prec_2]}
		   \bigoperSAND_{a_2\in\ccomp_2} \attr{a_2}
		   \bigg)
		   \Bigg)
		   \right)
		\hop
		\by{distr.\ of $\operSAND$ on $\operAND$ and on $\operOR$,
		    and distr.\ of $\operAND$ on $\operOR$}
		\hop
		&= \bigoperOR_{\poset[\attack_1][\prec_1]\in\dsem{v_1}}
		   \bigoperOR_{\poset[\attack_2][\prec_2]\in\dsem{v_2}}
		   \mathlarger{\Bigg(}
		   \bigoperAND_{\ccomp_1\in\Hasse[\attack_1][\prec_1]}
		   \bigoperAND_{\ccomp_2\in\Hasse[\attack_2][\prec_2]}
		   \cdots\\
		   &\cdots
		   \Bigg(
		   \bigoperSAND_{a_1\in\ccomp_1} \attr{a_1}
		   ~\operSAND~
		   \bigoperSAND_{a_2\in\ccomp_2} \attr{a_2}
		   \Bigg)
		   \mathlarger{\Bigg)}
		\hop
		\by{asoc.\ of $\operAND$ and $\operOR$, and
		    \Cref{lemma:dsem}.\ref{lemma:dsem:disjoint}}
		\hop
		&= \bigoperOR_{\poset\in\{\poset[\attack_1\uplus\attack_2]
		                                [\prec_1\uplus\prec_2]
		               \mid \poset[\attack_i][\prec_i]\in\dsem{v_i}\}~}
		   \bigoperAND_{~\ccomp_1,\ccomp_2\in\Hasse}
		   \!\!\!\cdots\\
		   &\cdots
		   \left(
		   \bigoperSAND_{a_1\in\ccomp_1} \attr{a_1}
		   ~\operSAND~
		   \bigoperSAND_{a_2\in\ccomp_2} \attr{a_2}
		   \right).
		\hop
	\end{flalign*}
	
	\vspace{-3ex}\noindent%
	Now consider the term $\big(\bigoperSAND_{a_1\in\ccomp_1} \attr{a_1}\big)
	\operSAND \big(\bigoperSAND_{a_2\in\ccomp_2} \attr{a_2}\big)$.
	The small operator $\operSAND$ in the middle states that any \BAS from $v_1$ must be to the left of that $\operSAND$, and any \BAS from $v_2$ must be to its right.
	That is connecting the connected components $\ccomp_1$ from $v_1$ with the $\ccomp_2$ from $v_2$, where the elements $a_1\in\ccomp_1$ precede the $a_2\in\ccomp_2$.
	So this ranges over edges $a_1\before a_2$ where nodes $a_i$ belong to the Hasse diagram \Hasse[\attack_i][\prec_i]: it links \Hasse[\attack_1][\prec_1] and \Hasse[\attack_2][\prec_2] via $\ATroot=\SAND(v_1,v_2)$.
	In other words, for $\ATroot$ those Hasse diagrams are joint, with edges that go from every $a_1$ in attacks of $v_1$ to each $a_2$ in attacks of $v_2$.
	Therefore:
	\begin{flalign*}
	\BU(\ATroot)
		&= \bigoperOR_{\poset\in\{
			 \poset[\attack_1\uplus\attack_2]
			      [\prec_1\uplus\prec_2\uplus{\attack_1\times\attack_2}]
		     \mid \poset[\attack_i][\prec_i]\in\dsem{v_i}\}~}
		   \cdots\\
		   &\cdots
		   \bigoperAND_{~\ccomp\in\Hasse}
		   \bigoperSAND_{a\in\ccomp} \attr{a}
		\hop
		\by{\Cref{lemma:dsem}.\ref{lemma:dsem:SAND}}
		\hop
		&= \bigoperOR_{\poset\in\dsem{\SAND(v_1,v_2)}}
		   \bigoperAND_{~\ccomp\in\Hasse}
		   \bigoperSAND_{a\in\ccomp} \attr{a}
		\hop
		\by{\Cref{def:DAT:metric}}
		\hop
		&= \metr{\SAND(v_1,v_2)}.
	\end{flalign*}
	
\end{proof}


\endgroup

}

\end{document}